\PassOptionsToPackage{nameinlink}{cleveref}
\documentclass[a4paper, UKenglish, numberwithinsect, thm-restate, cleveref, final]{lipics-v2021}
\usepackage{lineno}
\nolinenumbers

\newcommand{\takeout}[1]{\empty}

\usepackage{ifthen}
\newboolean{spellingmode} \setboolean{spellingmode}{false}
\newboolean{cropmargins} \setboolean{cropmargins}{false}
\newboolean{arxivmode} \setboolean{arxivmode}{false}

\usepackage[T1]{fontenc}
\usepackage{xcolor}
\usepackage{amsmath}
\usepackage{amssymb}
\usepackage{booktabs}
\usepackage{amsthm}
\usepackage{mathtools}
\usepackage{xspace}
\usepackage{stmaryrd}
\usepackage{microtype}
\usepackage{cite}
\usepackage{dsfont}
\usepackage{etoolbox}

\usepackage{ifdraft}
\ifthenelse{\boolean{arxivmode}}{\setboolean{cropmargins}{true}}{}

\ifspellingmode
\fi

\usepackage{tikz}
\usetikzlibrary{cd,calc,automata,positioning}

\tikzcdset{scale cd/.style={every label/.append style={scale=#1},
    cells={nodes={scale=#1}}}}

\usepackage{hyperref}
\usepackage{twshowkeys}

\title{Algebraic Language Theory with Effects}
\makeatletter
\hypersetup{
  final,
  hidelinks,
  pdftitle={\@title},
  pdfauthor={Lenke, Milius, Urbat, Wißmann},
}
\makeatother

\newcommand{\appendixproof}[2][Proof of]{%
  \subsection*{#1\ \autoref{#2}}%
  \addcontentsline{toc}{subsection}{#1\ \autoref{#2}}%
}
\newcommand{\detailsfor}[1]{%
  \appendixproof[Details for]{#1}%
}

\EventEditors{Keren Censor-Hillel, Fabrizio Grandoni, Joel Ouaknine, and Gabriele Puppis}
\EventNoEds{4}
\EventLongTitle{52nd International Colloquium on Automata, Languages, and Programming (ICALP 2025)}
\EventShortTitle{ICALP 2025}
\EventAcronym{ICALP}
\EventYear{2025}
\EventDate{July 8--11, 2025}
\EventLocation{Aarhus, Denmark}
\EventLogo{}
\SeriesVolume{334}
\ArticleNo{168}

\author{Fabian Lenke}{Friedrich-Alexander-Universität Erlangen-Nürnberg, Germany\and\url{https://www8.cs.fau.de/people/henning-urbat/}}{henning.urbat@fau.de}{https://orcid.org/0000-0001-5890-9485}{}
\author{Stefan Milius}{Friedrich-Alexander-Universität Erlangen-Nürnberg, Germany\and\url{https://www.stefan-milius.eu}}{henning.urbat@fau.de}{https://orcid.org/0000-0002-2021-1644}{}
\author{Henning Urbat}{Friedrich-Alexander-Universität Erlangen-Nürnberg, Germany\and\url{https://www8.cs.fau.de/people/henning-urbat/}}{henning.urbat@fau.de}{https://orcid.org/0000-0002-3265-7168}{}
\author{Thorsten Wißmann}{Friedrich-Alexander-Universität Erlangen-Nürnberg, Germany\and\url{https://thorsten-wissmann.de}}{thorsten.wissmann@fau.de}{https://orcid.org/0000-0001-8993-6486}{}

\authorrunning{F.\ Lenke, S.\ Milius, H.\ Urbat and T.\ Wißmann}

\Copyright{Fabian Lenke, Stefan Milius, Henning Urbat and Thorsten Wißmann}

\relatedversiondetails{Full version with appendix}{https://arxiv.org/abs/2410.12569}
\funding{Fabian Lenke and Henning Urbat are supported by Deutsche Forschungsgemeinschaft (DFG, German Research Foundation) -- project number 470467389. Stefan Milius is supported by Deutsche Forschungsgemeinschaft (DFG, German Research Foundation) -- project number 517924115.}
\def\resettheorembrackets{
\def\theorembracketopen{(}
\def\theorembracketclose{)}
}
\resettheorembrackets
\makeatletter
\def\@spopargbegintheorem#1#2#3#4#5{\trivlist
      \item[\hskip\labelsep{#4#1\ #2}]{#4{\theorembracketopen}#3{\theorembracketclose}\@thmcounterend\ }#5}
\makeatother

\newcommand{\resetCurThmBraces}{%
  \gdef\curThmBraceOpen{(}%
  \gdef\curThmBraceClose{)}}
\resetCurThmBraces
\newcommand{\removeThmBraces}{%
  \gdef\curThmBraceOpen{}%
  \gdef\curThmBraceClose{}}
\resetCurThmBraces

\newenvironment{notheorembrackets}{\removeThmBraces}{\resetCurThmBraces}

\usepackage{etoolbox}
\patchcmd{\thmhead}{(#3)}{\curThmBraceOpen #3\curThmBraceClose }{}{}

\theoremstyle{plain}
\newtheorem{lem}[theorem]{Lemma}

\newtheorem{game}[theorem]{Game}
\newtheorem*{openproblem}{Open Problem}

\theoremstyle{definition}
\newtheorem{defn}[theorem]{Definition} 
\newtheorem{expl}[theorem]{Example} 
\newtheorem{rem}[theorem]{Remark} 
\newtheorem{notn}[theorem]{Notation} 
\newtheorem{assumption}[theorem]{Assumption}
\crefname{expl}{Example}{Examples}
\crefname{defn}{Definition}{Definitions}

\newcommand{\finite}{fg-carried\xspace}

\renewcommand{\S}{\mathcal{S}}
\newcommand{\C}{\ensuremath{\mathcal{C}}}
\newcommand{\parfun}{\rightharpoonup}

\newcommand{\T}{\ensuremath{\mathbb{T}}\xspace}
\newcommand{\Id}{\ensuremath{\mathsf{Id}}}

\newcommand{\nil}{\ensuremath{\varepsilon}}

\newcommand{\seq}{\ensuremath{\mathbin{\text{\upshape;}}}}
\DeclareMathOperator{\Syn}{Syn}
\newcommand{\id}{\ensuremath{\mathsf{id}}}

\newcommand{\lst}{\mathsf{ls}}
\newcommand{\rst}{\mathsf{rs}}
\newcommand{\M}{\ensuremath{\mathcal{M}}\xspace}

\newcommand{\D}{\ensuremath{\mathcal{D}}}
\newcommand{\Kl}{\ensuremath{\mathcal{K}\!\ell}}
\newcommand{\Alg}{\ensuremath{\mathsf{Alg}}}

\newcommand{\N}{\ensuremath{\mathds{N}}}
\newcommand{\R}{\ensuremath{\mathds{R}}}
\newcommand{\Set}{\ensuremath{\mathsf{Set}}\xspace}

\newcommand{\epito}{\ensuremath{\twoheadrightarrow}}
\newcommand{\xto}{\xrightarrow}

\newcommand{\mult}{\mathbin{\boldsymbol{\cdot}}}

\def\colim{\qopname\relax m{colim}}
\newcommand{\Pfp}{\ensuremath{\mathcal{P}_{\mathrm{f}}^{+}}}
\newcommand{\Cp}{\ensuremath{\mathcal{C}}}

\newcommand{\xiCp}{\ensuremath{h^\#}}

\newcommand{\fp}{\mathsf{fp}}

\newcommand{\Cfp}{\C_\fp}

\newcommand{\kseq}{\ensuremath{\mathrel{\fatsemi}}}

\newcommand{\A}{\ensuremath{\mathcal{A}}\xspace}
\newcommand{\set}[2][]{%
  \ifthenelse{\equal{#2}{}}{%
    \ensuremath{{#1\emptyset}}%
  }{%
    \ensuremath{{#1\{#2#1\}}}%
  }%
}

\newcommand{\rn}[2]{
    \tikz[remember picture,baseline=(#1.base)]\node [inner sep=0] (#1) {$#2$};%
}

\newsavebox{\kleisliarrow}
\savebox{\kleisliarrow}{%
\begin{tikzpicture}[
      baseline=(arrow.base),
      inner sep=0mm,
      outer sep=0mm,
      ]
      \node[draw=none,
      anchor=base,
      overlay,
      inner sep=0,
      outer sep=0,
      minimum height=1em,
      ] (arrow) {$\phantom{\longrightarrow}$};

      \coordinate (circle pos) at
        ($ (arrow.south) !.68! (arrow.north)$);
      \begin{scope}[even odd rule,overlay]
        \clip  (circle pos) circle (0.17em)
           (arrow.north west) rectangle (arrow.south east);
      \node[draw=none,
      anchor=base,
      inner sep=0,
      outer sep=0,
      ] (arrow) {$\longrightarrow$};
    \end{scope}
    \draw[fill=none,overlay]
      ($ (arrow.south) !.68! (arrow.north)$)
      circle (0.15em);
    \draw[use as bounding box,draw=none] (arrow.north west) rectangle (arrow.south east);
  \end{tikzpicture}}

\newcommand{\kleislito}{\ensuremath{\mathbin{\usebox{\kleisliarrow}}}}

\newsavebox{\kleislidot}
\savebox{\kleislidot}{%
\begin{tikzpicture}[baseline=0pt,outer sep=0pt]
    \draw[fill=white,solid] (0,0) circle (1.5pt);
  \end{tikzpicture}}

\tikzstyle{kleisli}=[
"{\usebox{\kleislidot}}"{red,anchor=center,font=\normalsize,pos=0.5},
outer sep = 1pt, 
]

\tikzstyle{shiftarr}=[
        rounded corners,%
        to path={--([#1]\tikztostart.center)
                     -- ([#1]\tikztotarget.center) \tikztonodes
                     -- (\tikztotarget)},
]

\tikzstyle{dfa}=[
  x=15mm,
  y=15mm,
  initial text={},
  initial distance=4mm,
  every initial by arrow/.append style={
    shorten >= 2pt,
  },
  transition/.style={
    ->,
    shorten >= 2pt,
    shorten <= 2pt,
  },
  state/.append style={
    minimum size=5mm,
    inner sep=0pt,
  },
  accepting/.append style={
    double distance=2pt,
  },
]

\tikzset{
  loop at/.style={
    loop,
    out=#1-30,
    in=#1+30,
    looseness=6,
    every node/.append style={
      anchor=#1-180,
    },
  },
}

\newcommand{\cat}[1]{\ensuremath{\mathcal{#1}}\xspace}
\newcommand{\Conv}{\ensuremath{\mathsf{Conv}}\xspace}
\newcommand{\conv}{\Conv}

\newcommand{\Mon}{\mathsf{Mon}}

\def\colim{\qopname\relax m{colim}}
\DeclareMathOperator{\syn}{Syn}
\DeclareMathOperator{\ma}{Min}
\DeclareMathOperator{\tr}{Tr}
\DeclareMathOperator{\ar}{ar}

\newcommand{\V}{\mathcal{V}}

\newcommand{\incl}{\hookrightarrow}
\newcommand{\epi}{\twoheadrightarrow}

\DeclareMathOperator{\ev}{\mathrm{ev}}

\newcommand{\tensor}{\mathbin{\otimes}}

\DeclareMathOperator{\supp}{supp}

\DeclareMathOperator{\rng}{rng}

  \newcommand{\drawrule}[4]{ 
    \begin{tikzpicture}[overlay,remember picture]
        \draw [-, dotted] ({#1}) -- ({#2});
        \draw [-, dotted] ({#3}) -- ({#4});
    \end{tikzpicture}
  }

\usepackage{fixme}
\fxuselayouts{inline,nomargin}
\fxusetheme{color}
\FXRegisterAuthor{fb}{afb}{FB}
\FXRegisterAuthor{sm}{asm}{SM}
\FXRegisterAuthor{hu}{ahu}{HU}
\FXRegisterAuthor{tw}{atw}{TW}

\numberwithin{equation}{section}

\begin{document}

\maketitle
\begin{abstract}
  Regular languages -- the languages accepted by deterministic finite automata -- are known to be precisely the languages recognized by finite monoids. This characterization is the origin of algebraic language theory. In this paper, we generalize the correspondence between automata and monoids to automata with generic computational effects given by a monad, providing the foundations of an \emph{effectful} algebraic language theory. We show that, under suitable conditions on the monad, a language is computable by an effectful automaton precisely when it is recognizable by~(1) an effectful monoid morphism into an effect-free finite monoid,
  and~(2) a monoid morphism into a monad-monoid bialgebra whose carrier is a finitely generated algebra for the monad, the former mode of recognition being conceptually completely new. Our prime application is a novel algebraic approach to languages computed by probabilistic finite automata. Additionally, we derive new algebraic characterizations for nondeterministic probabilistic finite automata and for weighted finite automata over unrestricted semirings, generalizing previous results on weighted algebraic recognition over commutative rings.
\end{abstract}

\section{Introduction}

The algebraic approach to finite automata rests on the observation that regular languages (the languages accepted by finite automata) coincide with the languages recognized by \emph{finite monoids}. This provides the basis for investigating complex properties of regular languages with semigroup-theoretic methods. A prime example is the celebrated result by McNaughton, Papert, and Schützenberger~\cite{sch65,mp71} that a regular language is definable in first-order logic over finite words iff its syntactic monoid (the minimal monoid recognizing the language) is aperiodic. Since the latter property is easy to verify, this implies decidability of first-order definability. Similar characterizations and decidability results are known for numerous subclasses of regular languages~\cite{pin86,dgk08,pin15}. The correspondence between automata and monoids has been generalized beyond regular languages, for example to $\omega$-regular languages~\cite{wilke91,pp04}, languages of words over linear orders~\cite{Bedon2005}, data languages~\cite{bs20}, tree languages~\cite{almeida90,SalehiSteinby07,bw08,blumensath20}, cost functions~\cite{colcombet09}, and weighted languages over commutative rings~\cite{reu80}.

In this paper, we show that the equivalence between automata and monoids can be established at the more general level of automata with generic computational effects given by a monad $\T$~\cite{sbbr13,gms14}. This class of automata forms a common generalization of a wide range of automata models such as  probabilistic automata~\cite{Rabin63}, nondeterministic probabilistic automata~\cite{hhos18}, weighted automata~\cite{dkv09}, and even pushdown automata and Turing machines~\cite{gms14}. We introduce two modes of algebraic recognition for $\T$-effectful languages, both of which are natural effectful generalizations of the classical recognition by finite monoids: recognition by~(1) \emph{$\T$-effectful monoid morphisms} into finite effect-free monoids, 
and~(2) ordinary monoid morphisms into finitely generated $\T$-algebras equipped with an additional monoid structure (plus optional compatibility conditions).
We subsequently identify suitable conditions on the monad $\T$ ensuring that~(1) and~(2) capture precisely the languages computed by finite $\T$-automata; this yields an \emph{effectful automata/monoid correspondence} (\Cref{thm:em-recognition,thm:kleisli-recognition}). Our results fundamentally exploit the double role played by monads in computation, namely as abstractions of both computational effects~\cite{moggi91} and algebraic theories~\cite{manes76}.

The investigation of algebraic recognition at the present level of generality has several benefits. First and foremost, the abstract perspective provided by generic effects naturally motivates and isolates conceptual ideas that would be easily missed for concrete instantiations of the monad $\T$. Notably, recognition mode (1) above is conceptually completely new, and mode (2) generalizes earlier work on algebraic recognition over commutative varieties~\cite{amu18} to arbitrary (non-commutative) effects. In this way, our theory leads to novel algebraic characterizations of several important automata models.

Our prime application is a characterization of languages computed by probabilistic finite automata (PFAs)~\cite{Rabin63}. PFAs extend classical deterministic finite automata
by probabilistic effects, and thus serve as a natural model of state-based computations that involve uncertainty or randomization. These arise in many application domains \cite{paz1971,DBLP:books/daglib/0020348}, as witnessed for instance by the wide range and success of probabilistic model checkers~\cite{KNP11,DBLP:journals/sttt/HenselJKQV22}. On the theoretical side, PFAs share some remarkable similarities with finite automata; in particular, machine-independent characterizations of PFA-computable languages in terms of probabilistic regular expressions~\cite{bgmz12,rs24} and probabilistic monadic second-order logic~\cite{weidner-12} are known. However, the fundamental algebraic perspective on finite automata has thus far withstood a probabilistic generalization. We fill this gap by establishing two different modes of \emph{probabilistic algebraic recognition} (\Cref{thm:pfas-vs-finite-monoids,thm:pfas-vs-fgc-convex-algebras}) for PFA-computable languages which instantiate the two above-mentioned modes, namely~(1) recognition by finite monoids via \emph{probabilistic monoid morphisms}; (2) recognition by \emph{convex monoids} carried by a finitely generated convex set.
The first mode emphasizes the effectful nature of probabilistic languages, while the second one relies on traditional universal algebra. These characterizations may serve as a starting point for the algebraic investigation of PFA-computable languages.

Further notable instances of the general theory include algebraic characterizations of nondeterministic probabilistic automata and weighted automata. For the latter, a weighted automata/monoid correspondence was only known for weights from a \emph{commutative ring}~\cite{reu80}; our version applies to general semirings and thus captures additional  types of weighted automata such as min-plus and max-plus automata~\cite[Ch.~5]{pin21_handbook} and transducers~\cite[Ch.~3]{pin21_handbook}.

\subparagraph{Related Work.}  The use of category theory to unify results for different
automata models has a long tradition\cite{g72,am74,at89}. \T-automata were first studied by Goncharov et al.~\cite{gms14} as an instance of effectful \emph{coalgebras}~\cite{rutten00,sbbr13}; they also fit into the framework of \emph{functor automata} by Colcombet and Petri\c{s}an~\cite{cp19}.
On the algebra side, Boja{\'n}czyk~\cite{Bojan15} used monads to unify notions of algebraic recognition. This abstract perspective
on algebraic language theory has led to a series of works, including a uniform theory of Eilenberg-type
correspondences~\cite{uacm17} and an abstract account of logical definability of languages~\cite{blumensath21}. However, while each of the above works studies either automata-based or algebraic recognition individually, formal connections between both approaches are far less explored at categorical generality. The only
results in this direction appear in the work of Ad\'amek et al.~\cite{amu18} on the automata/monoid correspondence in commutative varieties. Our approach takes the step to non-commutative monads and on the way develops the entirely new concept of effectful language recognition.

\section{Algebraic Recognition of Regular Languages}
\label{sec:regular}
To set the scene for our algebraic approach to effectful languages, we recall the classical
correspondence between finite automata and finite monoids as recognizers for regular languages~\cite{pin15,rs59}. Let us settle some notation used in the sequel:

\begin{notn} Fix a finite alphabet $\Sigma$, and let $\Sigma^*$ denote the set of finite words over $\Sigma$, with empty word $\varepsilon \in \Sigma^{*}$. We put $1 = \set{\ast}$ and $2 = \set{\bot,\top}$. A map $x \colon 1\to X$ is identified with the element $x(\ast) \in X$, which by abuse of notation is also denoted by \(x \in X\). Maps $p\colon X\to 2$ (\emph{predicates}) are identified  with subsets of $X$; in particular, languages are presented as predicates $L\colon \Sigma^*\to 2$. We denote the composite of two maps $f\colon X\to Y$, $g\colon Y\to Z$ by $f\seq g\colon X\to Z$ (note the order!), and the identity map on $X$ by $\id_X\colon X\to X$.  We use $\mapsto$ to define anonymous functions.
  Lastly, $X\to Y$ denotes the set of all maps from $X$ to $Y$.
\end{notn}
\subparagraph{Finite Automata.} A \emph{deterministic finite automaton} (\emph{DFA}) $\A = (Q,i,\delta,o)$ consists of a finite set $Q$ of \emph{states} and maps representing an \emph{initial state}, \emph{transitions}, and \emph{final states}:
\begin{equation*}
  i\colon
  \begin{tikzcd}[cramped]
    1
    \arrow{r}
    & Q,
  \end{tikzcd}
  \qquad
  \delta\colon
  \begin{tikzcd}[cramped]
    Q\times \Sigma
    \arrow{r}
    & Q,
  \end{tikzcd}
  \qquad
  o\colon
  \begin{tikzcd}[cramped]
    Q
    \arrow{r}
    & 2,
  \end{tikzcd}
\end{equation*}

Let $\bar\delta\colon \Sigma\to (Q\to Q)$, $\bar\delta(a) = \delta(-, a)$, denote the curried  form of $\delta$, and define the \emph{iterated transition map} $\bar\delta^* \colon \Sigma^{*} \rightarrow (Q \rightarrow Q)$ and the \emph{language} $L \colon \Sigma^{*} \rightarrow 2$ \emph{computed by \A} by
\begin{equation}
  \bar \delta^{*}(\varepsilon) = \id_{Q}, \quad \bar \delta^{*}(wa) = \bar \delta^{*}(w) \seq \bar \delta^{*}(a) \quad \text{ and } \quad L(w) = i \seq \bar \delta^{*}(w) \seq o \quad \text{ for } w \in \Sigma^{*}.\label{eq:bar-delta}
\end{equation}
%
%
\subparagraph{Monoids.} A \emph{monoid} $(M,\mult ,e)$ is a set $M$ equipped with an
associative multiplication $ M\times M\xto{\mult } M$ and a neutral
element $e\in M$; that is, the equations
$(x\mult y) \mult z = x\mult(y\mult z)$ and $e\mult x=x=x\mult e$
hold for all $x,y,z\in M$.
A map \(h \colon N \rightarrow M\) between monoids \((N, \mult, n)\) and \((M, \mult, e)\) is a \emph{monoid morphism} if $h(n)=e$ and $h(x\mult y)=h(x)\mult h(y)$ for all $x,y\in N$.
\begin{expl}\label{ex:monoids}
\begin{enumerate}[(1)]
\item For every set $X$, the set $X\to X$ of endomaps forms a monoid with multiplication given by composition $\seq$ and neutral element $\id_X\colon X\to X$.
  \item The set $\Sigma^*$ of words, with concatenation as multiplication  and neutral element $\nil$, is the \emph{free monoid}
  on $\Sigma$: for every monoid $(M,\mult,e)$ and every map $h_0\colon \Sigma\to M$, there
  exists a unique monoid morphism $h\colon \Sigma^*\to M$ such that $h_0(a)=h(a)$ for all
  $a\in \Sigma$. The morphism $h$, the \emph{free extension} of~$h_0$, is given by
  $h(a_1\cdots a_n)=h_0(a_1)\mult\cdots\mult h_0(a_n)$ for $a_1,\ldots,a_n\in \Sigma$.  For
  instance, the map $\bar \delta^*\colon \Sigma^*\to (Q\to Q)$ defined in~\eqref{eq:bar-delta}
  is the free extension of $\bar \delta\colon \Sigma\to (Q\to Q)$.
\end{enumerate}
\end{expl}
Monoids enable an algebraic notion of language recognition. A monoid $M$ \emph{recognizes} the language $L\colon \Sigma^*\to 2$ if there exists a monoid morphism $h\colon \Sigma^*\to M$ and a predicate $p\colon M\to 2$ such that $L = h \seq p$.
Monoid recognition captures precisely the regular languages:
\begin{notheorembrackets}
\begin{theorem}[{\cite[Thm.\ 1]{rs59}}]
  For every language $L\colon \Sigma^*\to 2$, there exists a DFA computing~$L$ iff there exists a finite monoid recognizing $L$.
\end{theorem}
\end{notheorembrackets}
\begin{proof}[Proof sketch]
If a DFA $\A = (Q,i,\delta,o)$ computes $L$, the monoid $Q\to Q$ recognizes~$L$ via the morphism $\bar\delta^*\colon \Sigma^*\to (Q\to Q)$ and predicate $p\colon (Q\to Q)\to 2$ defined by $p(f) = i \seq f \seq o\in 2$.

Conversely, given a finite monoid $(M,\mult,e)$ that recognizes $L$ via a monoid morphism $h\colon \Sigma^*\to M$ and a predicate $p\colon M\to 2$, we can turn $M$ into a DFA $\A = (M, e, \delta, p)$ computing~$L$ with transitions defined by
\(
    \delta(m,a) = m\mult h(a).
\)
\end{proof}

\section{Algebraic Recognition of Probabilistic Languages}
\label{sec:fpa}

Before we present the correspondence of automata and monoids on the categorical level of
general effectful automata in \cref{sec:effectful}, we whet the reader's appetite by illustrating the important special case of probabilistic languages. These are languages computed
by \emph{probabilistic finite automata}~\cite{Rabin63}, whose computational effects are
finite probability distributions.  All results in this section are instances of 
those in \cref{sec:effectful}; however, for the convenience of the reader we provide sketches
of the concrete arguments our general proofs instantiate to.


\subsection{Probability Distributions and Probabilistic Channels}\label{sec:dist-monad}
A \emph{finite probability distribution} on a set $X$ is a map $d\colon X\to [0,1]$ whose \emph{support} $\supp(d):=\{x\in X \mid d(x)\neq 0\}$ is finite and which satisfies $\sum_{x\in X}d(x) = 1$. A distribution can be represented as a finite formal sum $\sum_{i\in I} r_i x_i$ where $x_i\in X$, $r_i\in [0,1], \sum_{i \in I} r_{i} = 1$ and $d(x)=\sum_{i\in I:\, x_i=x} r_i$ for $x\in X$. The set of all distributions on $X$ is denoted by \(\D X\).

A map of the form \(f \colon X \to \D Y\), denoted by \(f \colon X \kleislito Y\), is a \emph{probabilistic channel} or \emph{Markov kernel} from \(X\) to \(Y\). Intuitively,  \(f\) is a map from \(X\) to \(Y\) that assigns to a given input~\(x\) the output \(y\) with probability \(f(x)(y)\).
We write \(X \kleislito Y\) for the set of all probabilistic channels from \(X\) to \(Y\).
Two probabilistic channels \(f \colon X \kleislito Y\) and \(g \colon Y \kleislito Z\) can be composed to yield a probabilistic channel \(f \kseq g \colon X \kleislito Z\) given by
$(f\kseq g)(x) = \big(z \mapsto
    \sum_{y\in Y}
    f(x)(y)
    \cdot g(y)(z)\big)$.
    The \emph{unit} at $X$ is the probabilistic channel \(\eta_{X} \colon X \kleislito X\) sending \(x \in X\) to the \emph{Dirac distribution} \(\delta_{x} \in \D X\) defined by \(\delta_{x} = (y \mapsto 1 \text{ if } x = y \text{ else } 0 \)).
    Using sum notation, composition of probabilistic channels is given by \((f \kseq g)(x) = \sum_{i \in I}\sum_{j \in J_{i}} r_{i}r_{ij}z_{ij}\), where \(f(x) = \sum_{i \in I} r_{i} y_{i}\) and \(g(y_{i}) = \sum_{j \in J_{i}} r_{ij} z_{ij}\), and the unit is \(\eta_{X}(x) = 1x\).
    Composition \(\kseq\) of probabilistic channels is associative and has \(\eta\) as an identity: \((f \kseq g) \kseq h = f \kseq (g \kseq h)\) and \(\eta_{X} \kseq f = f \kseq \eta_{Y}\) for \(f \colon X \kleislito Y, g \colon Y \kleislito Z\) and \(h \colon Z \kleislito W\).
    A channel \(f \colon X \kleislito 2\) is a \emph{probabilistic predicate} \(f \colon X \rightarrow [0, 1]\) on \(X\), where we identify \(\D 2\) with the unit interval.
    A map \(f \colon X \rightarrow Y\) induces the \emph{pure}  channel \(f \seq \eta_{Y} \colon X \kleislito Y\), which we also denote by $f \colon X \rightarrow Y$ by abuse of notation.

\subsection{Probabilistic Automata}
Probabilistic automata, due to Rabin~\cite{Rabin63},  generalize deterministic automata. Transitions are no longer given by a unique successor state $\delta(q,a)$ for every state $q$ and input $a$, but a probability distribution over possible successor
states. Accordingly, probabilistic automata compute \emph{probabilistic languages}, which are simply probabilistic predicates $\Sigma^*\kleislito 2$. Formally:

\begin{defn}\label{def:pfa} A \emph{probabilistic finite automaton} (\emph{PFA}) $\A =
  (Q,i,\delta,o)$ consists of a finite set~$Q$ of \emph{states} and the following channels for
  an \emph{initial distribution}, the \emph{transition distributions}, and an \emph{acceptance
    predicate}, respectively:
\[
  i\colon
  \begin{tikzcd}[cramped]
    1
    \arrow[kleisli]{r}{}
    & Q,
  \end{tikzcd}
  \qquad
  \delta\colon
  \begin{tikzcd}[cramped]
    Q\times \Sigma
    \arrow[kleisli]{r}{}
    & Q,
  \end{tikzcd}
  \qquad
  o\colon
  \begin{tikzcd}[cramped]
    Q
    \arrow[kleisli]{r}{}
    & 2.
  \end{tikzcd}
\]
We denote by $\bar\delta\colon \Sigma\to (Q\kleislito Q)$ the curried form of $\delta$ and
define the \emph{iterated transition map} $\bar \delta^{*} \colon \Sigma^{*} \rightarrow (Q
\kleislito Q)$ and the \emph{language} $L \colon \Sigma^{*} \kleislito 2$ \emph{computed by
  \A}, respectively, by
\begin{equation*}
  \bar \delta^{*}(\varepsilon) = \eta_{Q}, \quad \bar \delta^{*}(wa) = \bar \delta^{*}(w) \kseq \bar \delta^{*}(a) \quad \text{ and } \quad L(w) = i \kseq \bar \delta^{*}(w) \kseq o \quad \text{ for } w \in \Sigma^{*}.\label{eq:bar-delta-kleisli}
\end{equation*}

\end{defn}
Comparing with the definition of DFAs in \Cref{sec:regular}, we see that DFAs are precisely
PFAs where $i$, $\delta$, $o$ are pure maps. 

\begin{rem} We think of $i(q)$ as the probability that $\A$ starts in state~$q$, of $\delta(q,a)(q')$ as the probability that $\A$ transitions from $q$ to $q'$ on input $a$, and of $o(q)\in \D2\cong [0,1]$ as the probability that $q$ is accepting. Unravelling the above definition, the language $L\colon \Sigma^*\kleislito 2$ computed by $\A$ is given by the explicit formula
  \begin{equation}\label{eq:pfa-lang}
    \textstyle
    L(w)
    =
    \sum_{\vec{q}\in Q^{n+1}} i(q_0) \cdot \big(\prod_{k=1}^n \delta(q_{k-1},a_k)(q_k)\big) \cdot o(q_n)\quad\text{for $w=a_1\cdots a_n$}.
  \end{equation}
The summand for $\vec{q}\in Q^{n+1}$ is the probability that, on input $w$, the automaton takes the path $\vec{q}$ and accepts $w$. Hence, $L(w)$ is the total acceptance probability.
\end{rem}

\begin{rem}
  \begin{enumerate}[(1)]
    \item Rabin's original notion of PFA~\cite{Rabin63} features an initial state and a set of
      final states, which amounts to restricting  $i\colon 1\kleislito Q$ and $o\colon
      Q\kleislito 2$ in \Cref{def:pfa} to pure channels. Except for the behaviour on the empty
      word, the two versions are expressively equivalent~\cite[Lemmas~3.1.1
      and~3.1.2]{bukharaev95}.
      
    \item The PFA model used here is often called \emph{reactive} in the literature, as opposed
      to \emph{generative} PFAs, whose transition map is of type
      \(Q \kleislito 1 + Q \times \Sigma\), computing \emph{stochastic languages}, which are
      (not necessarily finite!)  distributions over \(\Sigma^{*}\). Generative PFA are not
      instances of \T-automata (\cref{def:effectful-automaton}) and therefore not considered in
      this paper.
  \end{enumerate}
\end{rem}

Our aim is to understand PFA-computable probabilistic languages in terms of recognition by algebraic structures, in the same way that finite monoids recognize regular languages. To this end, we introduce two modes of probabilistic algebraic recognition and prove that they capture precisely the PFA-computable languages.

\subsection{Recognizing Probabilistic Languages by Finite Monoids}
\label{sec:alg-rec-prob-hom}
For our first mode of probabilistic algebraic recognition, we stick with finite monoids as recognizing structures, and only add probabilistic effects to the recognizing monoid morphisms.
First we need some auxiliary machinery:

\begin{definition}\label{def:xi-lambda-pi}
  For all finite sets $X$ and all sets $Y,Z$ we define the maps
  \begin{align}
    \xi_{X,Y}&\colon \D (X\to Y) \to (X\to \D Y),
    &
    \xi_{X,Y}(d)(x) &\textstyle\/= \big(y\mapsto
    \sum_{f\colon X\to Y, f(x) = y}
    d(f)\big),\label{eq:xi}\\
    \lambda_{X,Y}&\colon (X\to \D Y)\to \D (X\to Y),
    &
    \lambda_{X,Y}(g)&\textstyle\/= \big(f\mapsto \prod_{x\in X} g(x)(f(x))\big),\label{eq:lambda}\\
 \pi_{Y,Z}&\colon \D Y\times \D Z \to \D(Y\times Z),& \pi_{Y,Z}(d,e)&= \big((y,z)\mapsto d(y)\cdot e(z)\big).\label{eq:pi}
  \end{align}
\end{definition}
Intuitively,
$\xi_{X,Y}(d)(x)(y)$ is the probability of picking some $f$
according to $d$ satisfying $f(x) = y$.  For a channel $g\colon X\kleislito Y$, the probability
of $f\colon X\to Y$ in the distribution $\lambda_{X,Y}(g)$ provides a measure of how compatible
$f$ is to $g$.  The map $\pi_{Y,Z}$ sends two distributions on $Y$ and $Z$ to their \emph{product
  distribution} on $Y\times Z$.
Well-definedness of $\xi$ and $\lambda$ and the next lemma can be shown by calculation; conceptually, they follow from $\D$ being an \emph{affine monad}.

\begin{lem}\label{lem:xi-surj}
We have $\lambda_{X,Y} \seq \xi_{X,Y}=\id$. In particular, the map $\xi_{X,Y}$ is surjective.
\end{lem}

\begin{notn}\label{not:pi}
Given channels
$f_1\colon X_1\kleislito Y_1$ and
$f_2\colon X_2\kleislito Y_2$,
we define the channel $f_{1} \bar\times f_{2} \colon X_{1} \times X_{2} \kleislito Y_{1} \times Y_{2}$ as the composite in diagram\ in \autoref{fig:kleisli-pairing}.
\end{notn}

\begin{figure}
  \def\currentMinipageHeight{20mm}
  \begin{minipage}[b][\currentMinipageHeight]{.3\textwidth}\centering
      \begin{tikzcd}[cramped,column sep=9.5mm]
        X_1\times X_2
        \arrow[kleisli]{rd}[overlay, swap]{f_1 \bar\times f_2}
        \arrow{r}[overlay]{f_1\times f_2}
        & \D Y_1\times \D Y_2
        \arrow{d}[kleisli, pos=0.44,overlay]{\pi_{Y_1,Y_2}}
        \\
        & Y_1\times Y_2
      \end{tikzcd}
      \caption{Definition of $\bar\times$}
      \label{fig:kleisli-pairing}
  \end{minipage}%
  \hfill%
  \begin{minipage}[b][\currentMinipageHeight]{.47\textwidth}\centering
    \begin{tikzcd}
      N\times N
      \arrow{r}{\mult}
      \arrow[kleisli]{d}[swap]{h\bar\times h}
      & N
      \arrow[kleisli]{d}[swap]{h}
      & 1
      \arrow{l}[swap]{n}
      \arrow{dl}{e}
      \\
      M \times M
      \arrow{r}{\mult}
      &
      M
    \end{tikzcd}
    \caption{Probabilistic monoid morphism}
    \label{fig:prob-hom}
  \end{minipage}%
  \hfill%
  \begin{minipage}[b][\currentMinipageHeight]{.21\textwidth}\centering
    \begin{tikzpicture}[shorten >=1pt,node distance=18mm,on grid,auto,baseline=(q_0.south)]
      \node[state, inner sep=2pt, minimum size=0pt] (q_0) {\(q_{0}\)};
      \node[state, accepting, inner sep=2pt, minimum size=0pt] (q_1) [right=of q_0] {\(q_{1}\)};
      \path[->]
        (q_0) edge node {\(a|\frac{1}{2}\)} (q_1);
      \path[->]
        (q_0) edge [loop,out=60,in=120,looseness=6] node[above] {\(a|\frac{1}{2}\)} (q_0)
        (q_1) edge [loop,out=60,in=120,looseness=6] node[above] {\(a|1\)} (q_1);
    \end{tikzpicture}%
    \caption{A PFA}
    \label{fig:examplePFA}
  \end{minipage}%
\end{figure}

\begin{defn}
A \emph{probabilistic monoid morphism} from a monoid $(N,\mult,n)$ to a monoid $(M,\mult,e)$ is a channel $h\colon N\kleislito M$ making the diagram in \autoref{fig:prob-hom} commute, where the maps $\mult$, $n$, $e$ are regarded as pure channels.
\end{defn}

\begin{rem}\label{R:ext}\twnote{Why does uniqueness hold?}
  The universal property of $\Sigma^*$ extends to the probabilistic case: For every monoid
  $(M,\mult,e)$ and every channel $h_0\colon \Sigma\kleislito M$, there exists a unique probabilistic
  monoid morphism $h\colon \Sigma^*\kleislito M$ with $h(a)=h_0(a)$ for all $a\in \Sigma$. It is given by
  \[\textstyle h(w)= \big(m \mapsto \sum_{m=m_1\mult \cdots \mult m_n} \prod_{i=1}^n
    h_0(a_i)(m_i)\big) \qquad\text{for $w = a_{1} \cdots a_{n} \in \Sigma^{*}$}.\]
\end{rem}

\begin{defn}\label{D:prob-acc}
  A monoid $M$ \emph{probabilistically recognizes} a probabilistic language
  $L\colon \Sigma^*\kleislito 2$ if there exists a probabilistic monoid morphism
  $h\colon \Sigma^*\kleislito M$ and a probabilistic predicate $p\colon M\kleislito 2$ such
  that $L = h\kseq p$.
\end{defn}

We stress that, in \Cref{D:prob-acc}, probabilistic effects only appear in the channels~$h$ and $p$, while the monoid $M$ itself is pure. 
At first sight, it may seem more natural to use `probabilistic monoids', with proper channels as neutral element $1\kleislito M$  and multiplication $M\times M\kleislito M$, as recognizers. Remarkably, we need not require this additional generality:

\begin{theorem}\label{thm:pfas-vs-finite-monoids} For every probabilistic language  $L\colon \Sigma^*\kleislito 2$, there exists a PFA computing~$L$ iff there exists a finite monoid probabilistically recognizing $L$.
\end{theorem}

\begin{proof}[Proof sketch]
Given a PFA $\A=(Q,\delta,i,o)$ computing $L$, the finite monoid $Q\to Q$ probabilistically recognizes $L$ via the probabilistic monoid morphism $h\colon \Sigma^*\kleislito (Q\to Q)$ that freely extends the probabilistic channel $\bar\delta\kseq \lambda_{Q,Q}\colon \Sigma\kleislito (Q\to Q)$, and the probabilistic predicate $p\colon (Q\to Q)\kleislito 2$ defined by $p(f)=i\kseq f\kseq o\in \D2$. Explicitly, the maps $h$ and $p$ are given by
\[\textstyle h(w) = \big(f \mapsto \sum_{f=f_1\seq \cdots \seq f_n} \prod_{i=1}^n \prod_{q\in Q} \delta(q,a_i)(f_i(q))\big)\quad\text{and}\quad  p(f) = \sum_{q\in Q} i(q)\cdot o(f(q)), \]
for $w=a_1\cdots a_n$. A lengthy calculation using \Cref{lem:xi-surj} shows that $L=h\kseq p$.

Conversely, if a finite monoid $(M,\mult,e)$ probabilistically recognizes $L$ via $h\colon \Sigma^*\kleislito M$ and $p\colon M\kleislito 2$, then the PFA $\A=(M,\delta,e,p)$ computes $L$, with \(\delta \colon M \times \Sigma \kleislito M\) defined by
\[\textstyle \delta(m,a) = \big(n \mapsto \sum_{m'\,:\,m\mult m'=n} h(a)(m')\big). \qedhere\]
\end{proof}

\begin{example}\label{ex:kleisli-rec}
  The probabilistic monoid morphism induced by the channel \(h_{0} \colon \{a\} \kleislito (\{0, 1\}, \lor, 0)\) sending \(a \mapsto \frac{1}{2}0 + \frac{1}{2}1 \), together with the pure predicate \(\id \colon \{0, 1\} \rightarrow \{0, 1\} \) recognizes the language \(L(a^{n}) = 1 - \frac{1}{2^{n}}\).
  The corresponding PFA due to \Cref{thm:pfas-vs-finite-monoids} is
  shown in \autoref{fig:examplePFA}.
\end{example}

\subsection{Recognizing Probabilistic Languages by Convex Monoids}
\label{sec:alg-rec-by-convex-monoids}
The presence of probabilistic effects in the recognizing morphisms places the above mode of probabilistic recognition outside standard universal algebra. In this section we develop an equivalent, purely algebraic approach based on the theory of {convex sets}.

A \emph{convex set}~\cite{stone-49} is a set \(X\) equipped with a family of binary operations \( +_{r}  \colon X \times X \rightarrow X\) (\(r \in [0, 1]\)) subject to following equations, where \(s' = r + s -rs \ne 0\) and \(r' = \frac{r}{s'}\): 
  \[ x +_{r} x = x,
    \quad
    x +_{0} y = y,
    \quad
    x +_{r} y = y +_{1-r} x,
    \quad
    x +_{r} (y +_{s} z) = (x +_{r'} y ) +_{s'} z
  \]
  A map \(f \colon X \rightarrow Y\) between convex sets  
  is \emph{affine} if \(f(x +_{r} x') = f(x) +_{r} f(x')\) for \(x, x' \in X\) and \(r \in [0, 1]\).

\begin{expl}\label{ex:convex-sets}
\begin{enumerate}[(1)]
\item The prototypical convex sets are convex subsets $X\subseteq \R^\kappa$ (for a cardinal~$\kappa$) with the operations $\vec x +_r \vec y := r\cdot \vec x
+ (1-r)\cdot \vec y$. Up to affine isomorphism, these are precisely the \emph{cancellative} convex sets~\cite{stone-49}, which are those satisfying
  \[
    x +_{r} y = x +_{r} z \; \Longrightarrow \; y = z
    \qquad
    \text{for all $x, y , z \in X$ and $r \in (0, 1)$.}
  \]
\item
The set $\D X$ of distributions on a set $X$ is a convex set with structure given by $d+_r e = \big(x \mapsto r \cdot d(x) + (1-r) \cdot e(x)\big)$ for \(d, e \in \D X\). This is the \emph{free convex set on $X$}: every map $h \colon X \rightarrow Y$ to a convex set $Y$ extends uniquely to an affine map $h^\#\colon \D X\to Y$ such that $h = \eta_X\seq h^\#$. Concretely, $h^{\#}$ can be defined by $h^\#(\sum_{i=1}^{n} r_ix_i)=h(x_{1}) +_{r_{1}} h^{\#}(\sum_{i=2}^{n}\frac{r_{i}}{1-r_{1}}x_{i})$.
\item\label{ex:convex-sets-channels} For all sets $X$ and $Y$, the set $X\kleislito Y$ forms a convex set with the operations $f+_r g = \big(x \mapsto  f(x)+_r g(x)\big)$ for $f,g\colon X\kleislito Y$, $x\in X$ and $r\in [0,1]$.
\end{enumerate}
\end{expl}

Reutenauer\ \cite{reu80} showed that finite-dimensional \emph{\R-algebras} -- real vector
spaces equipped with a compatible monoid structure -- precisely recognize rational power series \(\Sigma^{*} \rightarrow \R\).
For algebraic recognition of probabilistic languages, we generalize \R-algebras to convex monoids, which are monoids with an additional convex structure that is respected by the multiplication:

\begin{defn}
  A \emph{convex monoid} is a convex set \(M\) equipped with a monoid structure $(M,\mult,e)$
  whose multiplication \(\mult \colon M \times M \rightarrow M\) satisfies
  \begin{equation}\label{eq:conv-mon}
    (m +_{r} m') \mult n
    =
    m \mult n +_{r} m' \mult n
    \quad \text{ and } \quad
    m \mult (n +_{r} n')
    =
    m \mult n +_{r} m \mult n'.
  \end{equation}
\end{defn}

\begin{expl}\label{ex:conv-mon}
\begin{enumerate}[(1)]
\item\label{ex:free-conv-mon} The convex set $\D \Sigma^*$ with multiplication
  $(\sum_{i\in I} r_iv_i)\mult (\sum_{j\in J} s_jw_j) = \sum_{i\in I,\,j\in J} r_is_jv_iw_j$
  and neutral element $\eta_{\Sigma^{*}}(\varepsilon)=1\varepsilon$ is the \emph{free convex
    monoid} on $\Sigma$: every map $h_0\colon \Sigma\to M$ to a convex monoid $M$ extends to
  a unique affine monoid morphism $h\colon \D\Sigma^*\to M$ such that $h(1a) = h_{0}(a)$.
  
\item\label{ex:dx-power} For every set $X$, the convex set $X\kleislito X$ from
  \Cref{ex:convex-sets}\ref{ex:convex-sets-channels} forms a convex monoid with channel
  composition $\kseq$ as multiplication and unit $\eta_X$ as neutral element.
\end{enumerate}
\end{expl}

\begin{defn}\label{D:conv-recog}
A convex monoid $M$ \emph{recognizes} a language $L\colon \Sigma^*\kleislito 2$ if there exists a monoid morphism $h\colon \Sigma^*\to M$ and an affine map $p\colon M\to \D 2$ such that $L=h\seq p$.
\end{defn}
For a correspondence between PFA-computable languages and convex monoids, we need to impose a suitable finiteness restriction on the latter. Finite convex monoids are not sufficient; instead, we shall work with a more permissive notion of finiteness:

\begin{defn}\label{D:conv-fg}
A convex set $X$ is \emph{finitely generated} if there exists an affine surjection $s\colon \D G\epito X$ for some finite set $G$. A convex monoid is \emph{\finite} if its underlying convex set is finitely generated.
\end{defn}

Intuitively, this definition says that \(X\) is the {convex hull} of a finite subset \(s[G] \subseteq X\): every element in \(X\) is a convex combination of the elements \(s(g), g \in G\).

\begin{expl}\label{ex:fg-carried}
\begin{enumerate}[(1)]
\item A convex subset of $\R^n$ is finitely generated iff it is a \emph{bounded convex polytope}~\cite{grunbaum03}, that is, it is compact and has finitely many extremal points.
\item\label{ex:x-kto-x}
For all finite sets $X$, the convex monoid $X\kleislito X$ from \Cref{ex:convex-sets}\ref{ex:convex-sets-channels} is fg-carried. This is witnessed by the map $\xi_{X,X}\colon \D (X\to X) \epito (X\to \D X)$ of~\eqref{eq:xi}, which is surjective by \Cref{lem:xi-surj}, and affine, since it is the free extension of the map $f\mapsto f\seq \eta_X$.
\end{enumerate}
\end{expl}

Fg-carried convex monoids give rise to our second algebraic characterization of PFAs:

\begin{theorem}\label{thm:pfas-vs-fgc-convex-algebras}
For every probabilistic language $L\colon \Sigma^*\kleislito 2$, there exists a PFA computing~$L$ iff there exists an \finite convex monoid recognizing $L$.
\end{theorem}

\begin{proof}[Proof sketch]
Given a PFA $\A=(Q,\delta,i,f)$ computing $L$, the \finite convex monoid $Q\kleislito Q$ from \Cref{ex:fg-carried}\ref{ex:x-kto-x} recognizes $L$ via the morphism  $\bar\delta^*\colon \Sigma^*\to (Q\kleislito Q)$ and the affine map $p\colon (Q\kleislito Q)\to \D 2$ given by $p(f)=i\kseq f\kseq o\in \D2$.

Conversely, suppose that $(M,\mult,e)$ is an \finite convex monoid (witnessed by an affine
surjection $s\colon \D Q\epito M$) recognizing $L$ via $h\colon \Sigma^*\to M$ and
$p\colon M\to \D2$. Define the PFA $\A_Q=(Q,\delta,i,o)$ where $o=\eta_Q\seq s \seq p$, the transition
distribution $\delta$ is chosen such that $s(\delta(g,a))=s(g)\mult h(a)$ for all $g\in Q$ and
$a\in \Sigma$, and the initial distribution $i$ is chosen such that $s(i)=e$; such choices
exist by surjectivity of $s$. Then $\A_Q$ computes $L$.
\end{proof}

A remarkable property of \finite convex monoids is that they always admit a finite presentation, despite not necessarily being finite themselves. To see this, let us recall some terminology from universal algebra. Given an equational class~$\V$ of algebras over a finitary signature $\Lambda$, a \emph{finite presentation} of an algebra $A\in \V$ is given by (1)~a finite set $G$ of \emph{generators}, (2)~a finite set $R$ of \emph{relations} $s_i=t_i$ ($i=1,\ldots,n$) where $s_i,t_i\in T_\Lambda G$ are $\Lambda$-terms in variables from $G$, and (3)~a surjective $\Lambda$-algebra morphism $q\colon T_\Lambda G\epito A$ satisfying $q(s_i)=q(t_i)$ for all $i$, subject to the universal property that every morphism $h\colon T_\Lambda G\to B$, where $B\in \V$ and $h(s_i)=h(t_i)$ for all $i$, factorizes through $q$. In particular, we have the notions of \emph{finitely presentable convex set} and \emph{finitely presentable convex monoid}.

 By a non-trivial result due to Sokolova and Woracek~\cite{sokolova-15}, finitely generated
 convex sets are finitely presentable (the converse holds trivially). This is the key to the
 following theorem; a categorical version of it is later proved in \Cref{thm:fg-carried-implies-fp-general}.

\begin{theorem}\label{thm:fg-carried-implies-fp}
Every \finite convex monoid is finitely presentable.
\end{theorem}
\begin{proof}[Proof sketch]
  Let $M$ be an \finite convex monoid. Choose a finite presentation
  $(G,R,q)$ of $M$ as a convex set. Since the multiplication of $M$ is
  fully determined by its action on $q[G]$ by \eqref{eq:conv-mon},
  this extends to a finite presentation of the convex monoid $M$ by
  adding a relation $g \mult g' = t_{g,g'}$ for each $g,g'\in G$,
  where $t_{g,g'}$ is any term in the signature of convex sets such
  that $q(t_{g,g'})=q(g)\mult q(g')$.  
\end{proof}

\noindent
We will see in \Cref{ex:syn-mon-not-fg-based-but-fp} that the converse of \Cref{thm:fg-carried-implies-fp} does not hold.

\subsection{Syntactic Convex Monoids}\label{sec:syn-conv-mon}
Compared to ordinary monoids, convex monoids are fairly complex structures. The benefit of the additional complexity is the existence of \emph{canonical} recognizers for probabilistic languages. Recall that every regular language $L\colon \Sigma^*\to 2$ has a canonical recognizer, the \emph{syntactic monoid} $\Syn(L)$. It is given by the quotient \(\Sigma^{*} \! / \!\! \approx_{L}\) of the free monoid $\Sigma^*$ modulo the \emph{syntactic congruence} $\approx_L\, \subseteq \Sigma^*\times \Sigma^*$, where $v \approx_{L} w$ iff $L(xvy) = L(xwy)$ for all $x,y\in \Sigma^*$. The projection $h_L\colon \Sigma^*\epito \Syn(L)$, sending $w\in  \Sigma^*$ to its congruence class $[w]$, recognizes~$L$ via \(p \colon \Syn(L) \rightarrow 2\) with \(p([w])= 1 \text{ iff } w \in L\). Moreover, $h_L$ factorizes through every surjective morphism $h\colon \Sigma^*\epito M$ recognizing $L$, thus $h_L$ is the `smallest' surjective morphism recognizing $L$.
For probabilistic languages, this notion generalizes as follows:

\begin{defn}\label{def:syn-conv-mon}
  Given a probabilistic language \(L \colon \Sigma^{*} \kleislito 2\),
  a \emph{syntactic convex monoid} of~$L$ is a convex monoid
  \(\syn(L)\) with a surjective affine monoid morphism
  \(h_{L} \colon \mathcal{D} \Sigma^{*} \epito \syn(L)\) such that
  $\eta_{\Sigma^*}\seq h_L$ recognizes $L$ and, moreover, $h_L$ factorizes
  through every surjective affine monoid morphism $h$ such that
  $\eta_{\Sigma^*}\seq h$ recognizes~\(L\):
  \medskip
     \[
    \begin{tikzcd}
      \Sigma^* \ar{r}{\eta_{\Sigma^*}} \ar[shiftarr={yshift=15pt}]{rrr}{L} & \mathcal{D} \Sigma^{*} \rar[two heads]{\forall h} \drar[two heads][swap]{h_{L}}
      & M \rar{p} \dar[dashed, swap]{\exists}
      & {[0, 1]}
      \\
      && \Syn(L) \urar[swap]{p_L} &
    \end{tikzcd}
    \]
\end{defn}
Syntactic structures for formal languages are well-studied from a categorical perspective~\cite{Bojan15,uacm17,amu18}. The following result is an instance of {\cite[Thm.~3.14]{amu18}}:
  \begin{theorem}\label{thm:syn-conv-mon}
    Every probabilistic language \(L \colon \Sigma^{*} \kleislito 2\) has a syntactic convex monoid, unique up to isomorphism. It is presented by generators \(\Sigma\) and relations given by the  \emph{syntactic congruence} $\approx_{L}\,\subseteq \D\Sigma^*\times \D\Sigma^*$ defined in \Cref{eq:synt-cong}. The maps \(h_{L} \colon \D \Sigma^{*} \rightarrow \syn(L)\) and \(p_{L} \colon \syn(L) \rightarrow [0, 1]\) are given by $h_L(\sum_i r_iv_i)=[\sum_i r_iv_i]$ and $p_L([\sum_{i}r_{i}v_{i}]) = \sum_{i}r_{i}L(v_{i})$.
    \begin{equation}
      \label{eq:synt-cong}\textstyle
      \sum_{i}r_{i}v_{i} \approx_L \sum_{j}s_{j}w_{j}\quad\text{iff}\quad \forall x, y \in \Sigma^{*} \colon \sum_{i}r_{i} L(xv_{i}y) = \sum_{j} s_{j} L(x w_{j} y)
    \end{equation}
  \end{theorem}
Alternatively, one can construct the syntactic convex monoid as the \emph{transition monoid} of the \emph{minimal $\D$-automaton} (see the full version~\cite{this-paper}), entailing a restriction on the convex structure of syntactic convex monoids:

\begin{theorem}\label{thm:syn-conv-mon-cancellative}
For every language $L\colon \Sigma^*\kleislito 2$, the convex set $\Syn(L)$ is cancellative.
\end{theorem}

Syntactic convex monoids are useful as a descriptional tool for characterizing (and potentially deciding) properties of languages in algebraic terms. Here is a simple illustration:

\begin{expl}\label{ex:commutative-lang}
  A probabilistic language $L\colon \Sigma^* \kleislito 2$ is \emph{commutative} if $L(a_1\cdots a_n) = L(a_{\pi(1)}\cdots a_{\pi(n)})$
  for all $a_1,\ldots, a_n\in \Sigma$ and all permutations $\pi$ of $\{1,\ldots,n\}$. One easily verifies that $L$ is commutative iff $\Syn(L)$ is a commutative convex monoid.
\end{expl}
Let us note that while every PFA-computable language is recognized by some \finite convex monoid (\Cref{thm:pfas-vs-fgc-convex-algebras}), its \emph{syntactic} convex monoid is generally not \finite:

\begin{expl}\label{ex:syn-mon-not-fg-based-but-fp} Consider the PFA from \Cref{ex:kleisli-rec} computing the probabilistic language \(L(a^{n}) = 1 - 2^{-n}\).
Its syntactic convex monoid $\Syn(L)$ is isomorphic to the half-open interval $(0,1]\subseteq\R$ with the usual convex structure and multiplication of reals; indeed, the map $i\colon \D\Sigma^*/{\approx_L} \to (0,1]$ given by $[\sum_k r_ka^{n_k}]\mapsto \sum_k r_k\cdot 2^{-n_k}$ is easily seen to be an isomorphism. Since the finitely generated convex subsets of $\R$ are closed intervals, $\Syn(L)$ is not \finite. However, despite \Cref{thm:fg-carried-implies-fp} not applying here, the convex monoid $(0,1]$ can be shown to be finitely presentable, with the finite presentation given by a single generator $a$ and a single relation $e+_{\frac{1}{3}} a\mult a = a$. The proof is somewhat intricate; see the full version~\cite{this-paper} for details.
\end{expl}

\begin{openproblem} Is $\Syn(L)$ finitely presentable for every PFA-computable language $L$?\end{openproblem}

\section{Algebraic Recognition of Effectful Languages}\label{sec:effectful}
We now turn to the main results of our paper: two novel modes of algebraic recognition for languages computed by effectful automata.
All results are parametric in the computational effect, which is  modelled by a monad satisfying a suitable condition. Our results instatiate to the
characterizations of PFA-computable probabilistic languages from \cref{sec:fpa}. Other instances of our results yield new algebraic characterizations for
languages recognized by weighted automata and automata that combine nondeterministic and
probabilistic branching.


\subsection{Monads}\label{sec:monads}

In the following, familiarity with basic category theory is assumed; see
Mac Lane~\cite{maclane} for a gentle introduction. 
We recall some concepts from the theory of monads~\cite{manes76} to fix our terminology and
notation. We write $\Set$ for the category of sets and functions. A \emph{monad}
$\T=(T,\eta,\mu)$ on $\Set$ is a triple consisting of an endofunctor $T\colon \Set\to\Set$ and
two natural transformations $\eta\colon \Id\to T$ (the \emph{unit}) and $\mu\colon TT\to T$
(the \emph{multiplication}), satisfying the laws $T\mu \seq \mu = \mu T \seq \mu$ and $T\eta \seq \mu =  \Id_T = \eta T \seq \mu$.

The \emph{Kleisli category} $\Kl(\T)$ has sets as objects, and a morphism from $X$ to $Y$, denoted $f\colon X\kleislito Y$, is a map $f\colon X\to T Y$. The composite of
  $f\colon X\kleislito Y$ and $g\colon Y\kleislito Z$ is denoted by $f\kseq g\colon X\kleislito
  Z$  and defined by $f\kseq g = f\seq Tg\seq \mu_Z$. The identity morphism on $X$ is the
  component $\eta_X\colon X\kleislito X$ of the unit. Intuitively, a Kleisli morphism is a
  function with computational effects given by the monad $\T$~\cite{moggi91}. A map $f\colon
  X\to Y$ is identified with the Kleisli morphism $f\seq \eta_Y\colon X\kleislito Y$; such
  Kleisli morphisms are said to be \emph{pure}, since they correspond to effect-free computations.
  Note that there is no need for operator precedence between \(\seq\) and \(\kseq\) since  \((f \seq g) \kseq h = f \seq (g \kseq h)\) and \((g \kseq h) \seq k = g \kseq (h \seq k)\) for all Kleisli morphisms \(g \colon X \rightarrow TY, h \colon Y \rightarrow TZ\) and pure maps \(f \colon W \rightarrow X, k \colon TZ \rightarrow U\).

  Further, a \emph{$\T$-algebra} $(A,a)$ consists of set $A$ (the \emph{carrier}) and a map
  $a\colon TA\to A$ (the \emph{structure}) satisfying the \emph{associative law} $Ta\seq a = \mu_A \seq a$ and the \emph{unit law}
  $\eta_A\seq a = \id_A$. A \emph{morphism} from $(A,a)$ to a $\T$-algebra $(B,b)$ (a
  \emph{$\T$-morphism} for short) is a map $h\colon A\to B$ such that $a\seq h = Th\seq b$. We
  write $\Alg(\T)$ for the category of $\T$-algebras and $\T$-morphisms. Products in $\Alg(\T)$
  are formed in $\Set$: the product algebra of $(A_i,a_i)$, $i\in I$, has the structure
  $T(\prod_i A_i)\xto{ \langle Tp_i\rangle_i} \prod_i TA_i
  \xto{\scalebox{0.6}{$\prod$}_i a_i} \prod_i A_i$, where 
  $p_i\colon \prod_i A_i\to A_i$ is the $i$th product projection.

The forgetful functor from $\Alg(\T)$ to $\Set$ has a left adjoint sending a set $X$ to the
\emph{free} $\T$-algebra  $\T X=(TX,\mu_X)$ over $X$. For every $\T$-algebra $(A,a)$ and every map
$h\colon X\to A$, there exists a unique $\T$-morphism $h^\#\colon \T X\to (A,a)$ such that $\eta_X
\seq h^\# = h$; that $\T$-morphism~$h^\#$ is the \emph{free extension} of $h$. The Kleisli category $\Kl(\T)$ is equivalent to a full subcategory of $\Alg(\T)$ via the embedding $X\mapsto \T X$ and~$h\mapsto h^\#$.

Monads provide a categorical view of universal algebra. Every finitary algebraic theory $(\Lambda,E)$,
given by a signature $\Lambda$ of finitary operation symbols and a set $E$ of equations between
$\Lambda$-terms, induces a monad $\T$ on $\Set$, where $T X$ is the carrier of the free
$(\Lambda,E)$-algebra (viz.~the set of all $\Lambda$-terms over $X$ modulo the equations in
$E$), and $\eta_X\colon X\to T X$ and $\mu_X\colon TT X\to TX$ are given by inclusion of
variables and flattening of terms over terms. Then $\T$-algebras bijectively correspond to
$\Lambda$-algebras satisfying all equations in $E$, that is, algebras of the \emph{variety}
specified by $(\Lambda,E)$. 
Monads $\T$ induced by finitary algebraic theories are precisely the \emph{finitary} monads, that is, those preserving directed colimits~\cite{adamek-rosicky-94}.

Every monad $\T$ has a canonical \emph{left strength}, the natural transformation
\[ 
  \lst_{X,Y}\colon X\times TY \to T(X\times Y)
  \qquad\text{defined by}\qquad
  \lst_{X, Y}(x,t) = T(y\! \mapsto\! (x,y))(t).
\]
 Its \emph{right strength} $\rst_{X,Y}\colon TX\times Y\to T(X\times Y)$ is defined analogously. The monad $\T$ is \emph{commutative} if $\rst_{X, TY} \kseq \lst_{X, Y} = \lst_{TX, Y} \kseq \rst_{X, Y}$ in $\Kl(\T)$.
 For every monad (be it commutative or not), we denote the left-hand composite  of this equation -- a \emph{double strength} -- by
\begin{equation}\label{eq:pi2}
\pi_{X,Y} :=
(\begin{tikzcd}[cramped,column sep=30]TX\times TY \ar{r}{\rst_{X,TY}} & T(X\times TY) \ar{r}{T\lst_{X,Y}} & TT(X\times Y) \ar{r}{\mu_{X\times Y}} & T(X\times Y)).
\end{tikzcd}
\end{equation}
Commutative finitary monads are precisely those induced by a \emph{commutative} algebraic
theory~$(\Lambda,E)$. This means that all operations commute with each other; for example, for
all binary operations \(\alpha, \beta \in \Lambda\), we have $\alpha(\beta(x_{1,2}, x_{1,2}), \beta(x_{2,1}, x_{2,2}))
  =
  \beta(\alpha(x_{1,1}, x_{2,1}), \alpha(x_{1,2}, x_{2,2}))
$,
 and similarly for every pair of operations of other (not necessarily equal) arities.

\begin{expl}\label{ex:monads}
In our applications we shall encounter the following monads:
\begin{enumerate}[(1)]
\item\label{ex:monads:D} The \emph{distribution monad} $\D$ sends a set $X$ to the set $\D X$
  of all finite probability distributions on $X$ (\Cref{sec:dist-monad}), and a map $f\colon
  X\to Y$ to the map \mbox{$\D f \colon \D X\to \D Y$} defined by $ \D f(d) =\big(y \mapsto
  \sum_{x\in X,\, f(x) = y} d(x)\big)$; in sum notation, $\D f(\sum_i r_ix_i)=\sum_i r_i
  f(x_i)$. Its unit \mbox{$\eta_X\colon X \to \D X$} is given by $\eta_X(x) = \delta_x$, and
  the multiplication $\mu_X\colon \D\D X \to \D X$ is defined by $\mu_X(e)=\big( x \mapsto \sum_{d\in \D X} e(d) \cdot d(x)\big)$. In sum notation,
  \(
    \eta_X(x)=1x\) and
    \(\mu_X(\sum_{i\in I} r_i(\sum_{j\in J_i} r_{ij}x_{ij}))
    =
    \sum_{i\in I}\sum_{j\in J_i} r_i r_{ij} x_{ij}\).
    The Kleisli category of $\D$ is the category of sets and probabilistic channels, and
    algebras for~$\D$ are precisely the convex sets~\cite{swirszcz-74}. The monad $\D$ is commutative; the natural transformation \eqref{eq:pi2} is concretely given by \eqref{eq:pi}.

  \item\label{ex:monads:C} The \emph{convex power set of distributions monad} $\C$ sends a set
    $X$ to the set $\C X$ of non-empty, finitely generated convex subsets of $\D X$, and a map
    $f\colon X\to Y$ to $\C f\colon \C X\to \C Y$ defined by $\C f (S) = \{ \D f (d)\mid d\in
    S\}$. Its unit $\eta_X\colon X\to \C X$ is $\eta_X(x)=\{\delta_x\}$, and the multiplication
    $\mu_X\colon \C\C X\to \C X$ is defined by 
    \[
      \mu_X(S) = 
      \textstyle\bigcup_{\Phi\in S} \{\sum_{U \in \supp(\Phi)} \Phi(U) \cdot d_U \mid \forall U\in
      \supp(\Phi)\colon d_U\in U \}
      \qquad\text{for every $S\in \C\C X$}.
    \]
    Algebras for $\C$ are precisely \emph{convex
      semilattices}~\cite{bonchi-22}, which are convex sets~$A$ carrying an additional
    \emph{semilattice} (i.e.~a commutative idempotent semigroup) structure $(A,+)$ satisfying
    $(x+y)+_r z = (x+_r z) + (y+_r z)$ for $x,y,z\in A, r\in [0,1]$. The monad $\C$ is not
    commutative.
  
\item\label{ex:monads:S} A \emph{semiring} is a set $S$ equipped with both the structure of a
  monoid $(S,\mult,1)$ and of a commutative monoid $(S,+,0)$ such that multiplication
  distributes over addition.  Every semiring $S$ induces a monad $\S$ sending a set $X$ to
  $\S X = \{ f\colon X\to S \mid f(x)\ne0 \text{ for finitely many $x\in X$} \}$, and a map
  $f\colon X\to Y$ to the map $\S f\colon \S X \to \S Y$ defined by
  $\S f(g) = \big(y\mapsto \sum_{f(x)=y} g(x)\big)$. The unit $\eta_X\colon X\to \S X$ is given
  by $\eta_{X}(x)=(y \mapsto 1 \text{ if } x=y \text{ else } 0)$, and the multiplication
  $\mu_X\colon \S\S X\to \S X$ is defined by
  $\mu_X(e)=\big(x \mapsto \sum_{d\in \S X} e(d) \cdot d(x)\big)$.  Algebras for~$\S$
  correspond precisely to \emph{$S$-semimodules}, that is, commutative monoids $(M,+,0)$ with an
  associative scalar multiplication $S\times M\to M$ that distributes over the additive
  structures of $S$ and $M$. The monad $\S$ is commutative iff the multiplication of $S$ is
  commutative. The distribution monad $\D$ is a submonad of $\S$ for the semiring $S=\R$ of
  reals with the usual operations.
\item\label{ex:monads:M} The monad $\M$ sends a set $X$ to $\M X=X^*$ (finite words over $X$),
  and a map $f\colon X\to Y$ to $\M f=f^*\colon X^*\to Y^*$ defined by $f^*(x_1\cdots
  x_n)=f(x_1)\cdots f(x_n)$. The unit $\eta_X\colon X\to X^*$ is given by $\eta_X(x)=x$ and
  the multiplication \mbox{$\mu_X\colon (X^*)^*\to X^*$} by flattening (concatenation) of
  words. Algebras for $\M$ correspond precisely to monoids.
\end{enumerate}
\end{expl}

\subsection{Automata with Effects}
Automata with computational effects in a monad were introduced in previous
work~\cite{sbbr13,gms14}. PFAs as in \Cref{def:pfa} are the instance where the monad is $\D$
and the output algebra is the free~$\D$-algebra $\D2$.

\begin{assumption}
  We fix a monad $\T=(T,\eta,\mu)$ on $\Set$ and a $\T$-algebra $O$.
\end{assumption}

\begin{defn}\label{def:effectful-automaton}
A \emph{finite $\T$-automaton} (\emph{$\T$-FA}) $\A=(Q,\delta,i,f)$ consists of a finite set
$Q$ of \emph{states} together with two $\Kl(\T)$-morphisms and a map as shown below:
\[
  i\colon
  \begin{tikzcd}[cramped]
    1
    \arrow[kleisli]{r}{}
    & Q,
  \end{tikzcd}
  \qquad
  \delta\colon
  \begin{tikzcd}[cramped]
    Q\times \Sigma
    \arrow[kleisli]{r}{}
    & Q,
  \end{tikzcd}
  \qquad
  o\colon
  \begin{tikzcd}[cramped]
    Q
    \arrow{r}{}
    & O.
  \end{tikzcd}
\]
They represent an \emph{initial state}, \emph{transitions}, and \emph{outputs}, respectively. We
define the curried version $\bar\delta\colon \Sigma\to (Q\kleislito Q)$ of $\delta$ and the
extended transition map $\bar\delta^*\colon \Sigma^*\to (Q\kleislito Q)$ just like in
\Cref{def:pfa}. The language $L\colon \Sigma^*\to O$ \emph{computed} by $\A$ is given by \(L(w)=i\kseq \bar\delta^*(w)\seq o^\#\).
\end{defn}

\begin{rem}\label{rem:t-aut-props}\begin{enumerate}[(1)]
\item 
Earlier works~\cite{sbbr13,gms14} model $\T$-automata as coalgebras $Q\to O\times (TQ)^\Sigma$,
and their semantics is defined via the final coalgebra for the functor $FX=O\times X^\Sigma$,
carried by the set $O^{\Sigma^*}$ of languages~\cite{rutten00}. \Cref{def:effectful-automaton}
is equivalent to the coalgebraic one -- modulo initial state and currying of transitions -- yet
better suited for our algebraic constructions.

\item $\T$-automata and their language semantics are also instances of the framework of functor automata by
  Colcombet and Petri\c{s}an~\cite{cp19} interpreted either in the Kleisli category for $\T$ for
  free output algebras $O = \T O_0$, or in the Eilenberg-Moore category for $\T$ for general output
  algebras and the state object being the free algebra $\T Q$.
  
\item Every $\T$-FA is equivalent to a $\T$-FA with a pure initial state: one may simply add a new pure initial state simulating the behaviour of a non-pure one~\cite[Rem.~1]{hhos18}. Hence, the essence of the effectful nature of $\T$-FAs lies in the transitions.
\end{enumerate}
\end{rem}

\begin{expl}\label{ex:t-fa}
  \begin{enumerate}[(1)]
    \item\label{ex:t-fa:pfa}
          $\D$-FAs with the convex output set $O=\D 2$ are precisely PFAs.
    \item\label{ex:t-fa:npfa}
      $\C$-FAs are \emph{nondeterministic probabilistic finite automata}
      ($\emph{NPFAs}$)~\cite{hhos18}. They combine nondeterministic and probabilistic branching
      and as such are closely related to Segala systems~\cite{segala95} and Markov decision
      processes~\cite[Ch.~36]{pin21_handbook}. We take the output convex semilattice
      $O=[0,1]_{\max}$ given by the interval $[0,1] \subseteq \R$
      with its usual convex structure and taking maxima as
      the semilattice operation. An NPFA $\A=(Q,\delta,i,o)$ consists of a finite
      state set $Q$ and maps
      $i\colon 1\to \C Q$, $\delta\colon Q\times \Sigma\to \C Q$, and $o\colon Q\to [0,1]$.
      If $\A$ is in state $q$ and receives the input letter $a$, then it chooses a
      distribution $d\in \delta(q,a)$ and transitions to the state $q'$ with the probability
      $d(q')$. The choices are made with the goal of maximizing the acceptance probability of
      the input. Formally, $\A$ computes the probabilistic language
      $L_{\max}\colon \Sigma^*\to [0,1]$ defined for $w=a_1\cdots a_n$ by
          \[w\,\mapsto\,\textstyle
          \max \{ \sum_{\vec{q}\in Q^{n+1}} \!d_{0}(q_{0}) \cdot (\prod_{k=1}^{n}d_{q_{k-1},k}(q_{k})) \cdot o(q_{n}) \mid d_{0} \in i, \forall q. \forall k. d_{q,k}\in \delta(q,a_k) \}.
          \]
          Under this semantics, NPFAs are more expressive than PFAs~\cite{hhos18}. Two alternative semantics emerge by modifying the output convex semilattice: for $O=[0,1]_{\mathsf{min}}$, the interval $[0,1]$ with the minimum operation, an NPFA computes the language $L_{\min}\colon \Sigma^*\to [0,1]$ of minimal acceptance probabilities. For $O=\C 2$, the convex semilattice of closed subintervals of $[0,1]$, it computes the language $L_{\mathrm{int}}\colon \Sigma^*\to \C 2$ sending $w\in \Sigma^*$ to the interval $[L_{\min}(w), L_{\max}(w)]$. See van Heerdt et al.~\cite{hhos18} for a detailed coalgebraic account of the different semantics.
    \item\label{ex:t-fa:wfa}
          $\S$-FAs with the output semimodule
          $O=\S 1 \cong S$, are precisely \emph{weighted finite automata} ($\!$\emph{WFAs})~\cite{dkv09} over the semiring $S$. A WFA computes a \emph{weighted language} (or \emph{formal power series}) \mbox{$L\colon \Sigma^*\to S$}, which is given by the formula \eqref{eq:pfa-lang}, with sums and products formed in $S$. Interesting  choices for the semiring \(S\) include:
          \begin{enumerate}[(1)]
            \item the semiring $\R$ of reals -- note that PFAs are a special case of WFAs over $\R$;
            \item the \emph{min-plus} or \emph{max-plus semiring} of natural numbers with the operation $\min$ (resp.\ $\max$) as addition and the operation $+$ as multiplication; WFAs then correspond to \emph{min-plus} and \emph{max-plus automata} computing shortest (resp.\ longest) paths~\cite[Ch.~5]{pin21_handbook}.
            \item the semiring of regular languages over an alphabet $\Gamma$ w.r.t.\ union and concatenation; WFAs are equivalent to \emph{transducers} computing relations $R\subseteq\Sigma^*\times \Gamma^*$~\cite[Ch.~3]{pin21_handbook}.
          \end{enumerate}
  \end{enumerate}
\end{expl}

\subsection{Recognizing Effectful Languages by Bialgebras}

We now introduce two modes of algebraic recognition for languages computed by $\T$-automata. In contrast
to the exposition in \Cref{sec:fpa} for the instance $\T = \D$, we swap the order of presentation and first consider the recognition by algebraic structures and
subsequently the recognition by effectful homomorphisms in \Cref{sec:kleisli-rec}, since from the categorical viewpoint
the latter is best understood in the context of the former.


\begin{defn}\label{def:t-m-bialgebra}
  A ($\T,\M$)\emph{-bialgebra} is a set $M$ equipped with both a $\T$-algebra
  structure~$a\colon TM \to M$
  and the structure of a monoid $(M,\mult,e)$. It is a \emph{$\T$-monoid} if the monoid multiplication
  $M\times M\xto{\mult} M$ is a \emph{$\T$-bimorphism}: for every $m\in M$ the maps
  $m\mult(-), (-)\mult m\colon M\to M$ are $\T$-endomorphisms on $(M,a)$.
\end{defn}

\begin{rem}\label{rem:t-mon}
Both $(\T,\M)$-bialgebras and $\T$-monoids admit a categorical view:
\begin{enumerate}[(1)]
  \item $(\T,\M)$-bialgebras correspond to algebras for the coproduct~\cite{ambl12} of the monads \T and \M.
\item\label{rem:t-mon:com} If $\T$ is commutative, then $\T$-monoids correspond to algebras for the composite monad~$\T\M$ induced by the canonical distributive law of $\M$ over $\T$~\cite[Thm~4.3.4]{mm07}. They also correspond to monoid objects in the closed monoidal category ($\Alg(\T),\otimes,\T 1)$ whose tensor product represents $\T$-bimorphisms (i.e.~$\T$-morphisms $A\otimes B\to C$ correspond to $\T$-bimorphisms $A\times B\to C$)~\cite{bn76,kock-71a,seal-2012}. The internal hom of $A,B\in \Alg(\T)$ is the algebra $[A,B]$ of all $\T$-morphisms from~$A$ to~$B$, viewed as a subalgebra of the product $B^{|A|}$. In particular, composition  $[A,B]\times [B,C]\xto{\seq} [A,C]$ is a $\T$-bimorphism, so $([A,A],\seq,\id_A)$ is a $\T$-monoid.
\end{enumerate}
\end{rem}

\begin{expl}\label{ex:tm-bialg}
\begin{enumerate}[(1)]
\item A $(\D,\M)$-bialgebra is a convex set carrying an additional monoid structure (without
  any interaction between the two structures). A $\D$-monoid is a convex monoid, that is, a
  $(\D,\M)$-bialgebra whose monoid multiplication satisfies~\eqref{eq:conv-mon}.
\item \label{ex:tx-pow-x}
For every set $X$, the set $X\kleislito X$ is a $(\T,\M)$-bialgebra with monoid structure given by Kleisli composition and $\T$-algebra structure given by the product $(\T X)^X$ of the free $\T$-algebra $\T X$. If $\T$ is commutative, then $X\kleislito X$ is a $\T$-monoid; in fact, it is isomorphic to the $\T$-monoid $[\T X,\T X]$ (\Cref{rem:t-mon}\ref{rem:t-mon:com}).
\end{enumerate}
\end{expl}

Both the notion of recognition of probabilistic languages by convex monoids and the respective finiteness condition (\cref{D:conv-recog,D:conv-fg}) are instances of the following: 

\begin{defn}\label{D:bialg}
  A $(\T,\M)$-bialgebra $M$ \emph{recognizes} the language $L\colon \Sigma^*\to O$ if $L=h\seq p$ for some monoid morphism $h\colon \Sigma^*\to M$ and some $\T$-morphism $p\colon M\to O$.
\end{defn}

\begin{defn}\label{D:fg-carried}
A $\T$-algebra $(A,a)$ is \emph{finitely generated} if there exists a surjective $\T$-morphism $\T G\epito (A,a)$ for some finite set $G$. A $(\T,\M)$-bialgebra is \emph{\finite} if its underlying $\T$-algebra is finitely generated.
\end{defn}

%
\begin{theorem}\label{thm:em-recognition}
Suppose that $X\kleislito X$ is a finitely generated $\T$-algebra for every finite set~$X$. Then for every language $L\colon \Sigma^*\to O$, the following are equivalent:
\begin{enumerate}[(1)]
\item There exists a $\T$-FA computing $L$.
\item There exists an \finite $(\T,\M)$-bialgebra recognizing $L$.
\end{enumerate}
If the monad $\T$ is commutative, then these statements are equivalent to:
\begin{enumerate}[(1)]
\setcounter{enumi}{2}
\item There exists an \finite $\T$-monoid recognizing $L$.
\end{enumerate}
\end{theorem}
\begin{rem}\label{R:prod-fg}
  The condition of the theorem is equivalent to finitely generated
  \T-algebras being closed under finite products.
\end{rem}
\begin{proof}[Proof sketch]
(1)$\Rightarrow$(2) Let $\A=(Q,\delta,i,o)$ be a $\T$-FA computing $L$. Then $L$ is recognized by the \finite $(\T,\M)$-bialgebra $Q\kleislito Q$ via the monoid morphism $\bar\delta^*\colon\Sigma^*\to (Q\kleislito Q)$ and the $\T$-morphism $p\colon (Q\kleislito Q) \to O$ which is given by $p(f)=i\kseq f \seq o^\#$. The proof that $p$ is a $\T$-morphism requires the initial state $i$ to be pure, which we may assume due to \Cref{rem:t-aut-props}.

\medskip\noindent (2)$\Rightarrow$(1) Let $M$ be an \finite $(\T,\M)$-bialgebra recognizing $L$ via the monoid morphism $h\colon \Sigma^*\to M$ and the $\T$-morphism $p\colon M\to O$. Since $M$ is finitely generated as a $\T$-algebra, there exists a surjective $\T$-morphism $s\colon \T Q\epito M$ for some finite set $Q$. Let $(M,\mult,e)$ denote the monoid structure, and let $s_0\colon Q\to M$ and $h_0\colon \Sigma\to M$ be the domain restrictions of $s$ and $h$. We construct a $\T$-FA $\A=(Q,\delta,i,o)$ computing $L$ as follows: we choose $i\colon 1\kleislito Q$ and $\delta\colon Q\times \Sigma \kleislito Q$ such that the first two diagrams below commute -- these choices exist because $s$ is surjective. Moreover, we define $o\colon Q\to O$ by $o:=s_0\seq p$ as in the third diagram.
\[
\begin{tikzcd}
1 \ar{r}{i} \ar{dr}[swap]{e} & TQ \ar[two heads]{d}{s} \\
& M
\end{tikzcd}
\qquad
\begin{tikzcd}[column sep = 30]
Q\times \Sigma \ar{rr}{\delta} \ar{d}[swap]{s_0\times \id} && TQ \ar[two heads]{d}{s} \\
M\times \Sigma \ar{r}{\id\times h_0} & M\times M \ar{r}{\mult} & M
\end{tikzcd}
\qquad
\begin{tikzcd}
Q \ar{d}[swap]{s_0} \ar{r}{o} & O \\
M \ar{ur}[swap]{p} &
\end{tikzcd}
\]
One can prove that the $\T$-FA $\A$ computes the language $L$.

\medskip\noindent
Now suppose that the monad $\T$ is commutative. Then (1)$\Rightarrow$(3) is shown like (1)$\Rightarrow$(2), adding the observation that the recognizing $(\T,\M)$-bialgebra $Q\kleislito Q$ is a $\T$-monoid  as in \Cref{ex:tm-bialg}\ref{ex:tx-pow-x}. The implication (3)$\Rightarrow$(2) is trivial.
\end{proof}

\begin{expl}\label{ex:xx-fg}
The condition of \Cref{thm:em-recognition} holds for $\T\in \{\D, \C, \S\}$ from \Cref{ex:monads}. We present for each finite set $X$ a finite set $M$ and a map $\xi_{0}\colon M\to (X\kleislito X)$ such that $\xi = \xi_{0}^\#\colon TM\epito (X\kleislito X)$ is surjective. The respective witnesses are given in \Cref{fig:witnesses}, where $X\parfun X$ denotes the set of all partial functions on $X$. For all three monads, the condition amounts to the observation that every effectful function $f\colon X\kleislito X$ can be built from (partial or total) \emph{effect-free} functions using $\T$-operations.
\begin{figure*}[t]
\centering
\begin{tabular}{@{}p{.6cm}p{1.2cm}p{1.4cm}p{8cm}@{}}
\toprule
$\T$ & $\T$-FA & $M$ & $\xi_{0}\colon M\to (X\kleislito X)$ \\
\midrule
$\D$ & PFA & $X\to X$ & $\xi_{0}(f)=\big(x \mapsto \eta_X(f(x))\big)$ \\
$\C$ & NPFA & $X\to X$ & $\xi_{0}(f)=\big(x \mapsto \eta_X(f(x))\big)$ \\
$\S$ & WFA &  $X\parfun X$ & $\xi_{0}(f)=\big(x \mapsto \eta_X(f(x))$ if $f(x)$ is defined, else $0\in \S X$ $\big)$ \\
\bottomrule
\end{tabular}
\caption{Witnesses for $X\kleislito X$ being (monoidally) finitely generated}\label{fig:witnesses}
\end{figure*}
\end{expl}

\begin{rem}
Using the abstract theory of syntactic structures~\cite{Bojan15,uacm17,amu18}, it is possible to associate a canonical algebraic recognizer to every language $L\colon \Sigma^*\to O$, namely its \emph{syntactic $(\T,\M)$-bialgebra} and, in the commutative case, its \emph{syntactic $\T$-monoid}. The syntactic convex monoid (\Cref{thm:syn-conv-mon}) is an instance of the latter.
\end{rem}


Next, we show that, under conditions on~$\T$, fg-carried $\T$-monoids are finitely presentable.

\begin{defn}
Let $\mathbb{S}$ be a monad. An $\mathbb{S}$-algebra is \emph{finitely presentable} if it is the coequalizer in $\Alg(\mathbb{S})$ of some pair $p,q\colon \mathbb{S} X\to \mathbb{S} Y$ of $\mathbb{S}$-morphisms for finite sets $X$ and $Y$.
\end{defn}

\begin{rem}
For the monad $\mathbb{S}$ associated to an algebraic theory $(\Lambda,E)$, an $\mathbb{S}$-algebra is finitely presentable iff its corresponding $\Lambda$-algebra in the variety defined by $(\Lambda,E)$ admits a finite presentation by generators and relations~\cite[Prop.~11.28]{arv10}. Instantiating $\mathbb{S}$ to the monad corresponding to the algebraic theory of $\T$-monoids, where $\T$ is finitary, we obtain a notion of \emph{finitely presentable $\T$-monoid}.
\end{rem}

\begin{theorem}\label{thm:fg-carried-implies-fp-general}
If $\T$ is finitary and commutative, and finitely generated $\T$-algebras are finitely presentable, then every \finite $\T$-monoid is finitely presentable.
\end{theorem}
\twnote{The proof of \autoref{thm:fg-carried-implies-fp-general} can be simplified using \cite[Thm.\ 3.5]{amsw19b}.}
Note that \cref{thm:fg-carried-implies-fp} is the instance of this result for $\T = \D$.


\subsection{Recognizing Effectful Languages by Finite Monoids}
\label{sec:kleisli-rec}
Our second mode of recognition of \T-FA-computable languages uses effectful monoid morphisms to finite monoids; recognition by probabilistic channels from \Cref{sec:alg-rec-prob-hom} is an instance.

\begin{notn}
Given $\Kl(\T)$-morphisms $f_i\colon X_i\kleislito Y_i$ ($i=1,2$), we define the Kleisli morphism
$f_1 \bar\times f_2\colon X_1\times X_2\kleislito Y_1\times Y_2$ as in \Cref{not:pi}, with $\T$ in lieu of $\D$.
\end{notn}

\begin{defn}\label{D:T-eff-mor}
A ($\T$-)\emph{effectful monoid morphism} from a monoid $(N,\mult,n)$ to a monoid $(M,\mult,e)$ is a $\Kl(\T)$-morphism $h\colon N\kleislito M$ such that the diagram in \autoref{fig:prob-hom} commutes. The monoid~$M$ \emph{effectfully recognizes} the language $L\colon \Sigma^*\to O$ if there exists an effectful monoid morphism $h\colon \Sigma^*\kleislito M$ and a map  $p\colon M\to O$ such that $L = h\seq p^\#$.
\end{defn}
The key condition on the monad $\T$ that makes effectful recognition work is isolated in \Cref{thm:kleisli-recognition}. This requires some technical preparation:

\begin{rem}\label{rem:mon-functor}
The natural transformations $\pi_{X,Y}\colon TX\times TY\to T(X\times Y)$ of~\eqref{eq:pi2} and
$\eta_X\colon X\to TX$ make $T$ a \emph{monoidal functor}, that is, $\pi$ and $\eta$ satisfy
coherence laws w.r.t.\ the natural isomorphisms $(X\times Y)\times Z \cong X\times (Y\times Z)$
and $1\times X\cong X\cong X\times 1$~\cite[Thm.~2.1]{kock70}. Consequently, $T$ preserves monoid structures: for
every monoid $(M,\mult,e)$, the set $TM$ forms a monoid with neutral element and multiplication,
respectively, defined by
\[ 1\xto{e} M\xto{\eta_M} TM\qquad\text{and}\qquad TM\times TM \xto{\pi_{M,M}} T(M\times
  M)\xto{T\mult} TM.\]
It is folklore that if $\T$ is commutative, then, for every $\T$-monoid $N$, the extension $h^{\#} \colon TM \rightarrow N$
of every monoid morphism $h \colon M \rightarrow N$ is also a monoid morphism. For non-commutative monads \(\T\), this is not true in general.
\end{rem}

\begin{defn}\label{def:mon-fin-gen}
 A $(\T,\M)$-bialgebra $N$ is \emph{monoidally finitely generated} if there exists a finite monoid $M$ and a monoid morphism $\xi_{0}\colon M\to N$ whose free extension $\xi = \xi_{0}^\#\colon TM\to N$ is a surjective monoid morphism.
\end{defn}

We are now ready to state the desired general result on effectful recognition. 
 Note that its condition implies that of \Cref{thm:em-recognition}; for a separating example see the full version~\cite{this-paper}. 

\begin{theorem}\label{thm:kleisli-recognition}
Suppose that for every finite set $X$ the $(\T, \M)$-bialgebra $X\kleislito X$ from \Cref{ex:tm-bialg}\ref{ex:tx-pow-x} is monoidally finitely generated. Then for every language $L\colon \Sigma^*\to O$, there exists a $\T$-FA computing $L$ iff there exists a finite monoid $\T$-effectfully recognizing $L$.
\end{theorem}
\begin{proof}[Proof sketch]
($\Rightarrow$) Suppose that $\A = (Q, \delta, i, f)$ is a $\T$-FA computing $L$. The proof of \Cref{thm:em-recognition} shows that $L$ is recognized by the $(\T,\M)$-bialgebra $Q\kleislito Q$; say $L=g\seq p$ for some monoid morphism $g\colon \Sigma^*\to (Q\kleislito Q)$ and some $\T$-morphism $p\colon (Q\kleislito Q)\to O$. (The specific choices of $g$ and $p$ in that proof are not relevant here.) Since $Q\kleislito Q$ is monoidally finitely generated, there exists a monoid morphism $\xi_{0}\colon M\to (Q\kleislito Q)$ such that~$M$ is finite and $\xi = \xi_{0}^{\#}\colon TM\epito (Q\kleislito Q)$ is a surjective monoid morphism. By the universal property of the free monoid $\Sigma^*$ and the surjectivity of $\xi$, we obtain a (not necessarily unique) monoid morphism $h\colon \Sigma^*\to TM$ with $g=h\seq \xi$ to fill the commutative diagram~\eqref{eq:tfa-to-finite-t-monoid}.
Moreover, one can show that $h\colon \Sigma^*\kleislito M$ is an \emph{effectful} monoid morphism. Then
the finite monoid $M$ effectfully recognizes $L$ via $h$ and $\xi_{0} \seq p$ because $L=g\seq p=h\seq (\xi_{0} \seq p)^\#$.
\begin{minipage}[c]{0.45\textwidth}
  \begin{equation}
    \label{eq:tfa-to-finite-t-monoid}
    \begin{tikzcd}
      \Sigma^* \ar{dr}[swap]{g} \ar[dashed, overlay]{r}{h} & TM \ar[two heads]{d}{\xi} \ar[overlay]{r}{(\xi_{0} \seq p)^\#} & O \\
      & Q\kleislito Q \ar{ur}[swap]{p} &
    \end{tikzcd}
  \end{equation}
\end{minipage}
\hfill
\begin{minipage}[c]{0.45\textwidth}
  \begin{equation}
    \label{eq:t-monoid-to-tfa}
    \begin{tikzcd}
      M \times \Sigma \ar[overlay]{r}{\eta_{M}\! \times h_{0}} \ar[kleisli]{d}{\delta} &
      TM \times TM \ar[kleisli]{d}{\pi_{M, M}} \\
      M &
      M \times M
      \ar[swap]{l}{\mult}
    \end{tikzcd}
  \end{equation}
  \hfill
\end{minipage}
($\Leftarrow$) Suppose that the finite monoid $(M,\mult, e)$
effectfully recognizes $L$ via $h\colon
\Sigma^*\kleislito M$ and $p\colon M\to O$. Then the $\T$-FA
$\A=(M,\delta,e,p)$ whose transition map
$\delta\colon M\times \Sigma\kleislito M$ is given by the composite\ \eqref{eq:t-monoid-to-tfa} (where $h_0(a) = h(a)$ for \(a \in \Sigma\)) computes $L$.\qedhere
\end{proof}

\begin{expl}\label{ex:xx-mon-fg}
The condition of \Cref{thm:kleisli-recognition} holds for $\T\in \{\D, \C, \S\}$, as the maps $\xi_{0}\colon M\to (X\kleislito X)$ from \Cref{fig:witnesses} witnessing that $X\kleislito X$ is finitely generated  also witness that $X\kleislito X$ is \emph{monoidally} finitely generated. We note that the verification that $\xi\colon TM\to (X\kleislito X)$ is a monoid morphism is non-trivial unless \T{} is commutative (see \Cref{rem:mon-functor}).
\Cref{prop:central-implies-mon-morphism} below simplifies the computations for
non-commutative \T, and also explains conceptually \emph{why} the non-commutative monads $\C$
and $\S$ satisfy the condition: for a finite set $X$, the $(\T, \M)$-bialgebra $X \kleislito X$ is generated by \emph{central} morphisms.
\end{expl}

 Recall that a Kleisli morphism \(f\colon X \kleislito Y\) is
 \emph{central}~\cite{carette-23,power-02} if for all 
 Kleisli morphisms \(f' \colon X' \kleislito Y' \), the following diagram commutes in $\Kl(\T)$.
\begin{equation}\label{eq:central}
  \begin{tikzcd}[row sep=0.5pt]
    & TY \times X' \ar[kleisli]{r}{\rst_{Y, X'}} & Y \times X' \ar{r}{1 \times f'} & Y \times TY' \ar[kleisli, pos=0.3]{rd}{\lst_{Y, Y'}}  & \\
    X \times X' \ar[pos=0.8]{ru}{f \times 1} \ar[swap, pos=0.9]{rd}{1 \times f'} & & & & Y \times Y' \\
    & X \times TY' \ar[kleisli]{r}{\lst_{X, Y'}} & X \times Y' \ar{r}{f \times 1} & TY \times Y' \ar[swap, kleisli,pos=0.3]{ru}{\rst_{Y, Y'}} &
  \end{tikzcd}
\end{equation}

\begin{proposition}\label{prop:central-implies-mon-morphism}
  If \(\xi_{0} \colon M \rightarrow (X \kleislito X)\) is a monoid morphism whose uncurried form \(X\times M \kleislito X\) is central, then its free extension \({\xi} \colon TM \rightarrow (X \kleislito X)\) is a monoid morphism.
\end{proposition}

\begin{rem}\label{rem:affine-implies-xx-mon-fg}
The witness $\xi_{0}\colon (X\to X)\to (X\kleislito X)$ defined by $f \mapsto f \seq \eta_{X}$
for the monad $\D$ from \Cref{def:xi-lambda-pi} works more generally for all \emph{affine
  monads}~\cite{kock-71} (i.e.~monads~$\T$ satisfying $T1\cong 1$). Hence, our \Cref{thm:em-recognition,thm:kleisli-recognition} apply to all affine monads $\T$.
\end{rem}

\subsection{Applications}
We proceed to instantiate \Cref{thm:em-recognition,thm:kleisli-recognition} to derive algebraic characterizations for several effectful automata models. For $\T=\D$, we recover \Cref{thm:pfas-vs-finite-monoids,thm:pfas-vs-fgc-convex-algebras} for PFAs:
\begin{theorem}\label{thm:pfa-alg}
A probabilistic language is PFA-computable iff it is $\D$-effectfully recognized by a finite monoid iff it is recognized by an \finite convex monoid.
\end{theorem}
 For $\T=\C$, we obtain a new algebraic characterization of NPFAs. Note that a \emph{$(\C,\M)$-bialgebra} is a convex semilattice with an additional monoid structure.
\begin{theorem}\label{thm:npfa-alg}
$\!$A probabilistic language is NPFA-computable iff it is $\C$-effectfully recognized by a finite monoid iff it is recognized by an \finite $(\C,\M)$-bialgebra.
\end{theorem}
Finally, $\T=\S$ gives a new algebraic characterization of WFAs. Note that an \emph{$(\S,\M)$-bialgebra} is an $S$-semimodule with an additional monoid structure. For a commutative semiring $S$, the notion of an $\S$-monoid is that of an \emph{(associative) $S$-algebra}.
\begin{theorem}\label{thm:wfa-alg}
An $S$-weighted language is WFA-computable iff it is $\S$-effectfully
recognized by a finite monoid iff it is recognized by an \finite
$(\S,\M)$-bialgebra, and, for a commutative semiring $S$, iff it is recognized by an \finite $S$-algebra.
\end{theorem}
For the case where $S$ is a commutative {ring} (i.e.~has additive inverses), we recover the
correspondence between WFAs and \finite $S$-algebras due to Reutenauer~\cite{reu80}. However,
\Cref{thm:wfa-alg} applies to every semiring, so it also yields novel algebraic characterizations of min- and max-plus automata and transducers (\Cref{ex:t-fa}\ref{ex:t-fa:wfa}) as special instances.

\section{Conclusions and Future Work}

We have developed the foundations of an algebraic theory of automata with generic computational
effects. Under suitable conditions on the effect monad \T{}, we characterized
$\T$-FA-computable languages by algebraic modes of effectful recognition. As special cases,
this entails the first algebraic characterizations of probabilistic automata and weighted
automata over unrestricted semirings. We proceed to give some prospects for future work.

Proving finite presentability of syntactic convex monoids is challenging even for simple probabilistic languages, as \Cref{ex:syn-mon-not-fg-based-but-fp} illustrates. Identifying conditions on the language ensuring this property is thus a natural question.

We aim to extend our theory beyond the category of sets. Of particular interest are \emph{nominal sets}, which would allow us to capture effectful (e.g.~probabilistic or weighted) register automata~\cite{bkl14}. The main technical hurdle lies in the construction $Q\mapsto (Q\to Q)$ of function spaces, which does not preserve orbit-finite (i.e.~finitely presentable) nominal sets. This might be overcome via the recently proposed restriction to \emph{single-use} functions~\cite{bs20,bns24}.

An orthogonal generalization of our theory concerns effectful languages beyond finite words. This requires switching from monoids to algebras for a monad $\mathbb{S}$~\cite{Bojan15}.
For example, taking the monad~$\mathbb{S}$ corresponding to \emph{\(\omega\)-semigroups}~\cite{pp04} could lead to algebraic recognition of effectful languages over \emph{infinite words}.
For a generalization of the recognition by $\T$-monoids, we expect that the interaction between the monads~$\mathbb{S}$ and~\(\T\) be given by some form of distributive law similar to the interaction of~\(\M\) and~\(\T\).

Lastly, our effectful automata/monoid correspondence could pave the way to a \emph{topological} account of effectful languages based on effectful profinite monoids. A first glimpse in this direction is given by Fijalkow~\cite{fijalkow-17} who analyzes the value-1 problem for PFAs in terms of a notion of \emph{free prostochastic monoid}. We aim to study effectful versions of the duality theory of profinite monoids~\cite{gehrke-13,bum24} and its applications, notably variety theorems~\cite{ggp08,uacm17}.

\bibliography{bibliography-new}

\clearpage
\appendix 

\section*{Appendix}

This appendix provides proofs and additional details omitted from the main text for space reasons. In \Cref{app:cats} we recall some concepts from category theory and establish a few auxiliary results that will be used in subsequent proofs. In \Cref{app:trans-mon} we explain the construction of syntactic convex monoids via transition monoids briefly mentioned in \Cref{sec:syn-conv-mon}. In \Cref{app:proofs} we present full proofs of all theorems in the paper, as well as more details for examples.

\section{Categorical Background}\label{app:cats}

\subsubsection*{The Monoidal Category of $\T$-Algebras}
We the closed monoidal structure on the category $\Alg(\T)$ for a commutative monad $\T$ on $\Set$ (\Cref{rem:t-mon}). We remark that the constructions shown here are not specific to $\Set$, but can be developed for monads on general symmetric monoidal closed categories with enough limits and colimits~\cite{kock-71a,kock-72,seal-2012}.

Given algebras $A,B,C\in \Alg(\T)$, a \emph{$\T$-bimorphism} from $A,B$ to $C$ is a map $f\colon A\times B\to C$ such that all maps $f(a,-)\colon B\to C$ and $f(-, b) \colon A \rightarrow C$ are $\T$-morphisms for $a\in A$ and $b \in B$. We write  $\Alg(\T)(A,B;C)$ for the set of all $\T$-bimorphisms from $A,B$ to $C$.

The \emph{tensor product} $A\otimes B$ of two $\T$-algebras $(A, a)$ and $(B, b)$ is defined as the coequalizer of the following pair of $\T$-morphisms, with $\pi$ given by \eqref{eq:pi2}:
\[
  \begin{tikzcd}[column sep = 50]
    T(TA \times TB)
    \ar[yshift=3]{r}{T \pi_{A,B}\,\seq\, \mu_{A\times B}}
    \ar[yshift=-3]{r}[swap]{a \times b}
    &
    T(A \times B) \ar[two heads]{r}{c_{A,B}} & A\otimes B
  \end{tikzcd}
\]
It has the following universal property: The map
\[t_{A,B} = \eta_{A\times B}\seq c_{A,B}\colon A\times B\to A\otimes B\]
is a \emph{universal $\T$-bimorphism} for \(A, B\), which means that (1) \(t_{A, B}\) itself is a bimorphism, (2) for every $\T$-bimorphism $f\colon A\times B\to C$ there exists a unique $\T$-morphism $h\colon A\otimes B\to C$ such that $f=t_{A,B}\seq h$.

The tensor product $\otimes$ makes $\Alg(\T)$ a symmetric monoidal closed category (with tensor unit $\T 1$). The right adjoint of $- \otimes A\colon \Alg(\T)\to \Alg(\T)$ is given by $[A,-]\colon \Alg(\T)\to \Alg(\T)$;
recall that the \emph{internal hom} $[A,B]=\Alg(\T)(A,B)$ is carried by the set of \(\T\)-morphisms from \(A\) to \(B\), with the subalgebra structure of the product algebra $B^{|A|}$, i.e.~with the $\T$-algebra structure defined pointwise from the structure of $B$. We thus have the following natural isomorphisms:
\[ \Alg(\T)(A,B;C) \cong \Alg(\T)(A\otimes B,C)\cong \Alg(\T)(A,[B,C]). \]
Since the tensor product has \(\T 1\) as tensor unit and satisfies \(\T(A \times B) \cong \T A \otimes \T B\), we have that the free algebra functor \(\T \colon \Set \rightarrow \Alg(\T)\) is strong monoidal.
Recall that a \emph{monoid} in a monoidal category $\C$ (with tensor product $\otimes$ and unit $I$) is given by an object $M$ with a unit $e\colon I\to M$ and multiplication $\mult\colon M\otimes M\to M$ satisfying diagrammatic versions of the associative and unit laws. A \emph{morphism} $h\colon (M,\mult,e)\to (M',\mult,e')$ of monoids is a $\C$-morphism $h\colon M\to M'$ that preserves that monoid structure. We write $\Mon(\C)$ for the category of monoids and their morphisms. Since in $\C=\Alg(\T)$ the tensor product $\otimes$ represents $\T$-bimorphisms and $I=\T 1$, the category $\Mon(\Alg(\T))$ is isomorphic to the category of $\T$-monoids, that is, $\T$-algebras with an additional monoid structure on the underlying set whose multiplication is a $\T$-bimorphism (\Cref{def:t-m-bialgebra}). Morphisms of $\T$-monoids are maps that are simultaneously $\T$-morphisms and monoid morphisms.

\subsubsection*{Finitary Functors}
A diagram $D\colon I\to \C$ in a category $\C$ is \emph{directed} if its scheme $I$ is a directed poset, that is, every finite subset of $I$ has an upper bound. A \emph{directed colimit} is a colimit of a directed diagram. A functor $F\colon \C\to \D$ is \emph{finitary} if it preserves directed colimits.

  \begin{proposition}\label{prop:pointwise-finitary-implies-finitary}
    A functor \(F \colon \C \times \cat D \rightarrow \cat E\) is finitary if and only if it is finitary in both arguments.
  \end{proposition}

\begin{proof}
We first note that for a directed set \(I\) the diagonal \(\Delta \colon I \rightarrow I \times I\) is \emph{final}, which means
for \((i, j) \in I \times I\) the slice \((i, j) \downarrow \Delta\) is non-empty and connected.
Indeed, for non-emptiness take an upper bound \(k\) of \(i, j\), then \((i, j) \le (k, k) = \Delta k\).
For connectivity assume that \((i, j) \le (k, k)\) and \((i, j) \le (k', k')\), then any upper bound \(k, k' \le k''\) joins \((i, j) \le \Delta k, \Delta k' \le \Delta k''\). Finality of $\Delta$ implies that any two diagrams $D,D'$ with $D= \Delta \seq D'$ have the same colimit~\cite[Section IX.3]{maclane}.

Now let \(F \colon \cat C \times \cat D \rightarrow \cat E\) be a functor, and suppose that
 \(F\) is finitary in both arguments. Let \(D = \langle D_{1}, D_{2} \rangle \colon I \rightarrow \cat C \times \cat D\) be a directed diagram. Then we can write \(D\) as the composite functor \(D = \Delta \seq (D_{1} \times D_{2})\). We prove that $F$ is finitary:
\begin{align*}
  \colim_{i \in I} F D(i) &= \colim_{i \in I} F (D_{1} \times D_{2}) \Delta (i) \\
                          &\cong \colim_{(i, j) \in I \times I} F (D_{1}(i), D_{2}(j))  && \text{\(\Delta\) final} \\
                          &\cong \colim_{i \in I} \colim_{j \in I} F(D_{1}(i),D_{2}(j)) && \text{colimits commute with colimits} \\
                          &\cong  F(\colim_{i \in I}D_{1}(i),\colim_{j \in I}D_{2}(j)) && \text{\(F\) finitary in both arguments}\\
                          &\cong  F(\colim_{(i,j) \in I \times I} (D_{1} \times D_{2})(i, j)) && \text{colimits in \(\cat C \times \cat D\) computed pointwise} \\
                          &\cong F(\colim_{i \in I}(D_{1} \times D_{2})\Delta(i)) && \text{\(\Delta\) final}\\
                          &\cong F(\colim_{i \in I}D(i))
\end{align*}
Conversely, if \(F\) is finitary, then for every $C\in \C$ the functor \(F(C, -)\) is finitary, being the composite of $F$ with the finitary functor $\langle C,\Id\rangle\colon \D\to \C\times \D$. Analogously, $F(-,D)$ is finitary for every $D\in \D$. So $F$ is finitary in both arguments.
\end{proof}

\subsubsection*{Finitely Presentable Algebras} The idea of presenting algebras by generators and relations can be generalized to the level of objects of arbitrary categories. An object \(X\) of a category $\C$ is \emph{finitely presentable} (\emph{fp}) if the hom functor \(\C(X, -)\colon \C\to \Set\) is finitary. In more explicit terms, this means that every morphism $f$ from $X$ into a directed colimit $\colim_{i \in I} C_i$ factorizes \emph{essentially uniquely} through some  colimit injections $\kappa_i\colon C_i \to \colim C_i$. This means that
\begin{enumerate}[(1)]
\item there exists $i \in I$ and a morphism $f'\colon X \to C_i$
  such that $f' \seq \kappa_{i} = f$, and
\item\label{cond:fp-2} given two such factorizations $ f' \seq \kappa_{i} =  f'' \seq \kappa_{i}$
  of $f$, there exists a connecting morphism $\kappa_{i,j}\colon C_i \to
  C_j$ in the diagram such that $f' \seq \kappa_{i,j}  = f'' \seq \kappa_{i,j} $.
\end{enumerate}
We denote the class of all finitely presentable object of \(\C\) by \(\Cfp\). 

For $\C=\Alg(\T)$, the categories of algebras for a finitary monad $\T$ on $\Set$ induced by an algebraic theory $(\Lambda,E)$, the abstract categorical notion of finite presentability coincides with the usual algebraic one. More precisely, the following statements are equivalent for every $A\in \Alg(\T)$, see e.g.~\cite[Prop.~11.28]{arv10}:
\begin{enumerate}[(1)]
\item $A$ is a finitely presentable object of $\Alg(\T)$.
\item $A$ is a coequalizer in $\Alg(\T)$ of a pair $p,q\colon \T X\to \T Y$ where $X,Y$ are finite.
\item $A$ has a finite presentation by generators and equations as an algebra of the variety specified by $\Sigma$ and $E$.
\end{enumerate} 

\begin{proposition}\label{prop:tensor-pres-fp}
  For every finitary commutative monad \(\T\) on \Set, the tensor product on $\Alg(\T)$ preserves finitely presentable algebras.
\end{proposition}
\begin{proof}
  We first show that every finitely presentable algebra \(A\) is also \emph{internally finitely presentable}:
  for every directed diagram \(C_{i}, i \in I\) of algebras with colimiting cocone \(\kappa_{i} \colon C_{i} \rightarrow \colim_{i}C_{i}\) the canonical bijection
  \begin{alignat*}{4}
    \varphi \colon&\quad&  \colim_{i}\,\, & \Alg(\T)(A, C_{i}) &\quad&\cong &\quad \Alg(\T)(A, & \colim_{i} C_{i})\\
     &&[f_{i} \colon & A \rightarrow C_{i}] &\quad &\mapsto &\quad ( f_{i} \seq \kappa_{i} \colon A \rightarrow & C_{i} \rightarrow \colim_{i}C_{i}),
  \end{alignat*}
  which is extending the cocone \(\Alg(\T)(A, \kappa_{i})\), is an algebra homomorphism. Consider the following diagram in \Set,
  where \(a\) denotes the respective algebra structures and \(\iota_{i} = \Alg(\T)(A, \kappa_{i}) ; \varphi^{-1} \colon \Alg(\T)(A, C_{i}) \rightarrow \colim_{i} \Alg(\T)(A, C_{i})\).
  \[
    \begin{tikzcd}
      T \Set(A, C_{i}) \ar[shiftarr={xshift=-40pt}, swap]{ddd}{a} \ar{rr}{T \Set(A, \kappa_{i})} & & T \Set(A, \colim_{i} C_{i}) \ar[shiftarr={xshift=50pt}]{ddd}{a}  \\
      T \Alg(\T)(A, C_{i}) \ar{u}{T(\incl)} \ar{d}{a}  \ar{r}{T \iota_{i}} & T(\colim_{i} \Alg(\T)(A, C_{i})) \ar{r}{T \varphi}  \ar{d}{a} & T \Alg(\T)(A, \colim_{i} C_{i}) \ar{d}{a} \ar{u}{T (\incl)}   \\
      \Alg(\T)(A, C_{i}) \dar[hook]{} \rar{\iota_{i}} & \colim_{i} \Alg(\T)(A, C_{i}) \rar{\varphi} & \Alg(\T)(A, \colim_{i} C_{i}) \dar[hook]{}    \\
      \Set(A, C_{i}) \ar{rr}{\Set(A, \kappa_{i})} & & \Set(A, \colim_{i} C_{i})  \\
    \end{tikzcd}
  \]
  Then the top and bottom squares commute by definition of \(\varphi\);
  the outer left  and right parts commute because the inclusions are \T-morphisms;
  the center left square by definition of the algebra structure on the filtered colimit.
  This shows that the center right square commutes when pre- and postcomposed with \(T\iota_{i}\) and the inclusion, respectively.
  The former are collectively epic and the latter is monic, thus we get the desired commutativity of the center right square.

This entails that $\otimes$ preserves finitely presentable algebras. Indeed, given finitely presentable algebras $A,B$, the tensor \(A \tensor B\) is finitely presentable because its hom functor preserves filtered colimits:
  \begin{align*}
      & \Alg(\T)(A \tensor B, \colim_{i} C_{i}) \\
       \cong\,\, & \Alg(\T)(A, \Alg(\T)(B, \colim_{i}C_{i})) \\
    \cong\,\, & \Alg(\T)(A, \colim_{i} \Alg(\T)(B, C_{i})) \\
    \cong\,\, & \colim_{i} \Alg(\T)(A, \Alg(\T)(B, C_{i})) \\
    \cong\,\, & \colim_{i}\Alg(\T)(A \tensor B, C_{i}) \qedhere
  \end{align*}
%
\end{proof}

\begin{lemma}\label{lem:ful-faith-fin-reflect-fp}
    Every fully faithful finitary functor \(U \colon \C \rightarrow \cat D\) reflects fp objects.
\end{lemma}

\begin{proof}
Let \(C\) be an object of \(\C\) such that \(UC\) is fp in \(\D\).
Then for every directed diagram \(C_{i}, i \in I\) in \cat C we have
\begin{align*}
  \cat{C}(C, \colim_{i} C_{i}) &\cong \cat{D}(UC, U\colim_{i} C_{i}) && \text{\(U\) fully faithful} \\
                               &\cong \cat{D}(UC, \colim_{i}UC_{i}) && \text{\(U\) finitary} \\
                               &\cong \colim_{i}\cat{D}(UC, UC_{i}) && \text{\(UC\) fp} \\
                               &\cong \colim_{i} \cat{C}(C, C_{i}) && \text{\(U\) fully faithful. }
\end{align*}
This proves that $C\in \Cfp$.
\end{proof}
Recall that an \emph{algebra} for an endofunctor $F\colon \C\to \C$ (short: an \emph{\(F\)-algebra}) is a pair $(A,a)$ of an object $A\in \C$ and a morphism $a\colon FA\to A$. A \emph{morphism} $h\colon (A,a)\to (B,b)$ of $F$-algebras is a morphism $h\colon A\to B$ of $\C$ such that $a\seq h= Fh\seq b$. We write $\Alg(F)$ for the category of $F$-algebras and their morphisms.

\begin{defn}
  Let \(\C\) be a category and let \(\D \in \{\Alg(\T), \Alg(F), \Mon(\C)\}\).
  We call an object of \(\D\) \emph{fp-carried} if its underlying \(\C\)-object is finitely presentable.
\end{defn}

\begin{proposition}\label{prop:fp-carrier-f-alg}
  Suppose that \(\C\) has directed colimits and that $F\colon \C \to \C$ is finitary and preserves finitely presentable objects.
  Then every fp-carried \(F\)-algebra is finitely presentable.
\end{proposition}
The analogous statement for $F$-\emph{co}algebras is known~\cite[Lemma~3.2]{ap04}, and the proof follows a similar pattern.
\begin{proof}
  Let \((A, a)\) be an \(F\)-algebra such that \(A\in \Cfp\). To prove that $(A,a)$ is finitely presentable as an algebra, let \((A_{i}, a_{i})\), $i \in I$ be a directed diagram in \(\Alg(F)\).
  The colimit is created by the forgetful functor \(\Alg(F) \rightarrow \C\) since \(F\) is finitary (see~\cite[Thm~4.6]{rutten00} for the dual result for coalgebras). Therefore, the colimit is given by  \((\colim A_{i}, a)\) for the unique \(F\)-algebra structure \(a \colon F(\colim A_{i}) \rightarrow \colim A_{i}\) such that the colimit injections \(\kappa_{i}\colon A_i \to \colim A_i\) are  algebra morphisms.
  Now let \(f \colon (A,a) \rightarrow (\colim A_{i}, a)\) be an algebra morphism. We have to show that it factorizes essentially uniquely through a colimit injection.
  Since \(A\) is finitely presentable in~\C, the underlying \C-morphism \(f \colon A \rightarrow \colim A_{i}\) factorizes in \(\C\) through a colimit injection as  \(f = f_{i} \seq \kappa_{i} \) for some morphism $f_i\colon A \rightarrow A_{i}\). We do not claim that \(f_{i}\) is an algebra morphism.
  But we see that $\kappa_i$ merges $a \seq f_i$ and $Ff \seq F a_i$ using that $f$ is an algebra morphism:
  \[
    (a \seq f_{i}) \seq  \kappa_{i} = a \seq f = F f \seq a  = F f_{i} \seq F \kappa_{i} \seq a = (F f_{i}  \seq F a_{i}) \seq \kappa_{i}.
  \]
  Since \(F\) preserves fp objects, we have that \(FA \in \Cfp\). Thus, there exists a connecting morphism \(a_{i,j}\) merging \(a \seq f_{i} \) and \(F f_{i} \seq a_{i} \).
  The composite \(f_{j} = f_{i} \seq a_{i,j} \) is a morphism of \(F\)-algebras; indeed, the outside of the diagram below commutes since so do the upper and lower parts as well as the right-hand square, whereas the left-hand square commutes when post-composed with $a_{i,j}$:%
  \[
    \begin{tikzcd}
      A
      \rar{f_{i}}
      \ar[shiftarr={yshift=20}]{rr}{f_{j}}
      &
      A_{i}
      \rar{a_{i,j}}
      &
      A_{j}
      \\
      FA
      \uar[swap]{a}
      \rar{F f_{i}}
      \ar[shiftarr={yshift=-20}]{rr}{F f_{j}}
      &
      F A_{i}
      \uar[swap]{a_{i}}
      \rar{F a_{i,j}}
      &
      F A_{j}
      \uar[swap]{a_{j}}
    \end{tikzcd}
  \]
  Moreover, $f_j$ is the desired factorization: $f_j \seq \kappa_j = f_i \seq a_{i, j} \seq \kappa_j = f_i \seq \kappa_i = f$. Essential uniqueness of the factorization is clear since it already holds in $\C$. This proves that $(A,a)$ is finitely presentable.
\end{proof}

The above proposition extends from functor algebras to monoids:

\begin{theorem}\label{thm:mon-fp}
  Let \((\C,\otimes,I)\) be a closed monoidal category with directed colimits such that $\otimes$ preserves finitely presentable objects and  and $I$ is fp.
  Then every fp-carried monoid is finitely presentable.
\end{theorem}
\begin{proof}
  Since the category \(\C\) is closed, the tensor preserves arbitrary colimits in both arguments and thus is finitary by \autoref{prop:pointwise-finitary-implies-finitary}.
  Therefore the functor \(FX = I + X\otimes X\) is finitary, and it is fp-preserving by assumption. (Note that fp objects are closed under finite coproducts in every category.)
  Consider the functor \(U \colon \Mon(\C) \rightarrow \Alg(F)\) that turns a monoid $(M,\mult,e)$ into an $F$-algebra $I+M\otimes M\xto{[e,\mult]} M$, and is identity on morphisms. If $M\in \Cfp$, then
   $U(M,\mult,e)$ is finitely presentable in \(\Alg(F)\) by \autoref{prop:fp-carrier-f-alg}.
  Since \(U\) is fully faithful and finitary, \autoref{lem:ful-faith-fin-reflect-fp} shows that \((M,\mult,e)\) is finitely presentable in \(\Mon(\C)\).
\end{proof}
By \Cref{prop:tensor-pres-fp} the condition of the above theorem is satisfied in the category of algebras for a finitary commutative monad. Therefore, we conclude:

\begin{corollary}\label{cor:fb-impl-fp}
  Let \(\T\) be a finitary commutative monad on \Set and let $M$ be a $\T$-monoid. If the $\T$-algebra underlying $M$ is finitely presentable, then $M$ is finitely presentable as a $\T$-monoid.
\end{corollary}

\section{Syntactic Convex Monoids and Transition Monoids}\label{app:trans-mon}
We present an alternative construction of the syntactic convex monoid of a probabilistic language (\Cref{def:syn-conv-mon}) as a \emph{transition monoid}. For this purpose we first need a probabilistic automaton model that admits \emph{canonical} minimal automata, which PFA do not. We instead consider deterministic automata in the category $\conv$ of  convex sets and affine maps.

\begin{defn}
  Let \(\Sigma\) be a finite alphabet.
  A \emph{convex automaton} consists of a convex state set \(Q\), an inital state \(q_{0} \in Q\), an affine output map \(o \in Q\to [0,1]\) and a transition map \(d \colon Q \times \Sigma \rightarrow Q\) such that every \(d(-, a) \colon Q \rightarrow Q\) is affine.
\end{defn}

Note that convex automata  are simply algebras for the endofunctor \(F = 1 + - \tensor \mathcal{D} \Sigma\) on \conv with an affine output map.
Equivalently, convex automata are pointed coalgebras for the endofunctor \(G = [0, 1] \times (-)^{\Sigma}\) on \conv.
The free convex monoid \(\mathcal{D} \Sigma^{*}\) naturally carries the structure of an \(F\)-algebra by right multiplication \(\mathcal{D} \Sigma^{*} \otimes \mathcal{D} \Sigma \cong \D (\Sigma^{*} \times \Sigma) \rightarrow \D \Sigma^{*}\) and with \(\eta_{\Sigma^{*}}(\varepsilon) \colon 1 \rightarrow \D \Sigma^{*}\) as initial state.
Dually, the convex set of all probabilistic languages $[0,1]^{\Sigma^*} \cong \conv(\mathcal{D} \Sigma^{*}, [0, 1])\) is a \(G\)-coalgebra with transition map the curryfication of
\[\conv(\mathcal{D} \Sigma^{*}, [0, 1]) \tensor \mathcal{D} \Sigma \tensor \mathcal{D} \Sigma^{*} \xrightarrow{\id \tensor m} \conv(\mathcal{D} \Sigma^{*}, [0, 1]) \tensor \mathcal{D} \Sigma^{*} \xrightarrow{\ev} [0, 1],\]
where \(m\) is monoid multiplication in \(\mathcal{D} \Sigma^{*}\).
The output map is evaluation at the empty word.

\begin{notheorembrackets}
  \begin{corollary} [{\cite{goguen-75,amu18}}]
    \begin{enumerate}[(1)]
      \item The $F$-algebra \(\mathcal{D} \Sigma^{*}\) is the initial \(F\)-algebra.
      \item The $G$-coalgebra \([0, 1]^{\Sigma^{*}}\) is the final \(G\)-coalgebra.
    \end{enumerate}
  \end{corollary}
\end{notheorembrackets}

A convex automaton \(A\) is \emph{reachable} if the unique \(F\)-algebra morphism \(\mathcal{D} \Sigma^{*} \rightarrow A\) is surjective and it is \emph{simple} if the unique \(G\)-coalgebra morphism \(A \rightarrow [0, 1]^{\Sigma^{*}}\) is injective; \(A\) is \emph{minimal} if it is both reachable and simple.
Given a probabilistic language \(L\) we can turn the initial \(F\)-algebra \(\mathcal{D}\Sigma^{*}\) and the final \(G\)-coalgebra \([0, 1]^{\Sigma^{*}}\) into convex automata by choosing \(L\) as output map and initial state, respectively,
to get the initial and final automata recognizing \(L\).
Then the \emph{minimal convex automaton} \(\ma(L)\) is then given by the image of the unique automaton morphism \(\mathcal{D} \Sigma^{*} \rightarrow [0, 1]^{\Sigma^{*}}\) (obtained via its surjective-injective factorization), see~\cite{goguen-75}.
The \emph{transition monoid} \(\tr(A)\) of a convex automaton \(A = (Q, d, i)\) is the image of the monoid homomorphism \(\mathcal{D} \Sigma^{*} \rightarrow \conv(Q, Q)\) induced by \(a \mapsto d(-, a)\).

\begin{notheorembrackets}
  \begin{corollary}[\cite{amu18}]\label{thm:synt-mon-trans-mon}
 Let $L\colon \Sigma^*\to [0,1]$ be a probabilistic language, then the syntactic convex monoid is the transition monoid of the minimal automaton of \(L\): \[\syn(L) \cong \tr(\ma(L)).\]
  \end{corollary}
\end{notheorembrackets}

\section{Omitted Proofs and Details}\label{app:proofs}

\detailsfor{def:xi-lambda-pi}
The map \(\xi_{X,Y}\colon \D(X \rightarrow Y) \rightarrow (X \rightarrow \D
Y)\) is the free extension \(\xi_{X, Y} = s^{\#}\) of \[s \colon (X
\rightarrow Y) \rightarrow (X \rightarrow \D Y), \quad f \mapsto f \seq
\eta_{Y}\] and therefore well-defined.

It is a standard result that the product distribution $\pi_{Y,Z}$ is a
probability distribution on $Y\times Z$. Iterating the product distribution
$|X|$-times for a fixed set $Y$ (i.e.~$Z := Y$) and
identifying such $|X|$-tuples with maps $f\colon X\to Y$ yields the
definition of $\lambda_{X,Y}$, which is therefore a probability distribution, too.

\appendixproof{lem:xi-surj}
The equation \(\lambda_{X, Y} ; \xi_{X, Y} = \id\) holds for every affine monad,
as shown in the \hyperref[proof:rem:affine-implies-xx-mon-fg]{Details for
\Cref{rem:affine-implies-xx-mon-fg}} (where the morphism \(\xi_{X, Y}\) is
denoted by \(h^{\#}\)).

\appendixproof[Proof Sketch for]{thm:fg-carried-implies-fp}

  [For a full proof, see the proof of the general \Cref{thm:fg-carried-implies-fp-general} for \(\T = \D\).]

  Let $M$ be an \finite convex monoid. Choose a finite presentation
  $(G,R,q)$ of $M$ as a convex set. Since the multiplication of $M$ is
  fully determined by its action on $q[G]$ by~\eqref{eq:conv-mon},
  this extends to a finite presentation of the convex monoid $M$ by
  adding a relation $g \mult g' = t_{g,g'}$ for each $g,g'\in G$,
  where $t_{g,g'}$ is any term in the signature of convex sets such
  that $q(t_{g,g'})=q(g)\mult q(g')$.

\detailsfor{ex:commutative-lang}
We prove that for every probabilistic language $L\colon \Sigma^* \kleislito 2$,
\[ \text{$L$ is commutative}\quad\text{iff}\quad \text{$\Syn(L)$ is commutative}.\]

  If \(\Syn(L)\) is commutative, then we have for all $a_1,\ldots,a_n\in \Sigma$ and all permutations $\pi$ of $\{1,\ldots,n\}$:
  \begin{align*}
L(a_1\cdots a_n) &= p_L(h_L(a_1\cdots a_n)) \\
&= p_L(h_{L}(a_{1}) \cdots h_{L}(a_{n})) \\
&= p_L(h_L(a_{\pi(1)}) \cdots h_L(a_{\pi(n)})) \\
&= p_L(h_L(a_{\pi(1)} \cdots a_{\pi(n)}))\\
& = L(a_{\pi(1)}\cdots a_{\pi(n)});
\end{align*}
the third equality uses commutativity of $\Syn(L)$. Thus, $L$ is commutative. 

Conversely, if \(L\) is commutative, then in its syntactic monoid all pairs of elements \([\sum_{i}r_{i}v_{i}], [\sum_{j}s_{j}w_{j}] \in \syn(L)\) satisfy
  \begin{align*}
    \textstyle
    \big[\sum_{i}r_{i} v_{i}\big]\big[\sum_{j}s_{j}w_{j}\big]
    &=
    \big[\sum_{i}\sum_{j} r_{i}s_{j} v_{i}w_{j}\big] \\
    &=
    \big[\sum_{j}\sum_{i} s_{j}r_{i}w_{j}v_{i}\big] \\
    &=
    \big[\sum_{j}s_{j} w_{j}\big] \big[\sum_{i}r_{i}v_{i}\big];
  \end{align*}
  the second equality uses commutativity of \(L\). Thus, $\Syn(L)$ is commmutative.

\appendixproof{thm:syn-conv-mon-cancellative}
  Note first that the full subcategory of $\Conv$ given by cancellative convex sets is closed under arbitrary products and convex subsets; this is because the cancellation property is defined via Horn implications of equations, so cancellative convex sets form a \emph{quasivariety}.
  Let \(L \colon \Sigma^{*} \rightarrow [0, 1]\) be a language.
  The minimal convex automaton \(\ma(L)\) of \(L\) is the image of the unique automaton morphism from the inital to the final convex automaton for $L$ (\Cref{app:trans-mon}).
  Its carrier \(A_{L}\) is cancellative since it is a convex subset of \([0, 1]^{\Sigma^{*}}\).
  By \autoref{thm:synt-mon-trans-mon} we get
  \[\syn(L) \cong \tr(\ma(L)) \incl \conv(A_{L}, A_{L}) \incl A_{L}^{|A_{L}|}.\]
 Thus, \(\syn(L)\) is cancellative, being a convex subset of a power of $A_L$.

\detailsfor{ex:syn-mon-not-fg-based-but-fp}
A convex monoid $M$ is {finitely presentable} iff there exists a finite set $G$ and a surjective affine monoid morphism $e\colon \D G^*\epito M$ whose kernel congruence $\{(\varphi,\varphi')\mid e(\varphi)=e(\varphi')\}$ is finitely generated. If $G=\{g_1,\ldots,g_m\}$ and the kernel is generated by the pairs $(\varphi_1,\varphi_1'),\ldots,(\varphi_n,\varphi_n')$, we say that $M$ is finitely presented by the generators $g_1,\ldots,g_m$ and equations $\varphi_1=\varphi_1',\ldots, \varphi_n=\varphi_n'$. We prove the following claim:
\begin{proposition}\label{prop:iso-01}
The convex monoid $((0, 1], \cdot, 1)$ is presented by a single generator $a$ and the equation $a^{0} +_{\frac{1}{3}} a^{2} = a^{1}$.
\end{proposition}
\begin{proof}
  The map \(\{a\} \mapsto (0, 1]\), $a \mapsto \frac{1}{2}$, extends to the affine monoid morphism \(E = E_{X} \colon \mathcal{D} a^{*} \rightarrow (0, 1]\) that sends a finite distribution \(\varphi\) over \(a^{*}\) to its expected value \(E_{X}\varphi\) with respect to the random variable \[X \colon a^{*} \rightarrow \R, \quad a^{n} \mapsto \frac{1}{2^{n}}, \qquad \text{i.e.} \qquad E(\varphi) = \sum_{n \in \N}\frac{\varphi(n)}{2^n}.\]
  It is easy to see that \(E\) is surjective:
  every \(x \in (0, 1]\) has a unique representation as \(x = \frac{r}{2^{n}} + \frac{1-r}{2^{n+1}}\) for unique \(n \in \N, r \in (0, 1]\).
  This representation gives a canonical preimage of \(x\) under \(E\) given by the distribution \(\varphi_{x} = a^{n} +_{r} a^{n+1} \in \mathcal{D} a^{*}\).
  We call \(\varphi_{x}\) the \emph{canonical representative} of \(x \in (0, 1]\).

To prove the proposition, it suffices to show that the kernel of $E$ coincides with the convex monoid congruence $\sim\,\subseteq \D a^* \times \D a^*$ generated by $a^{0} +_{\frac{1}{3}} a^{2} \sim a^{1}$:
\[ \forall \varphi,\varphi'\in \D a^*.\; \varphi \sim \varphi'\iff E(\varphi)=E(\varphi'). \]
For the $\Rightarrow$ direction we only need to check that the pair  $(a^{0} +_{\frac{1}{3}} a^{2}, a^{1})$ is contained in the kernel congruence of $E$, which follows from
\[ E(a^0+_{\frac{1}{3}} a^2) = \frac{1}{3}\cdot 1 + \frac{2}{3}\cdot \frac{1}{4} = \frac{1}{2} = E(a^1).\]
For the $\Leftarrow$ direction it suffices to show that $\varphi \sim \varphi_{E(\varphi)}$ for every $\varphi\in \D a^*$; then $E(\varphi)=E(\varphi')$ implies $\varphi \sim \varphi_{E(\varphi)} \sim \varphi_{E(\varphi')} \sim \varphi'$ as required.
  For the proof of \(\varphi \sim \varphi_{E(\varphi)}\), let us first note that
 \(a^{n} +_{\frac{1}{3}} a^{n+2} \sim a^{n+1}\) for all $n$ since $\sim$ is closed under multiplication.
 We claim that for all \(\lambda\) with \(3\lambda + \sum_{n}r_{n} = 1\) we have
  \[(\sum_{i \in \N \setminus \{n, n+1, n+2\}} r_{i}a^{i}) + (r_{n} + \lambda)a^{n} + r_{n+1}a^{n+1} + (r_{n+2} + 2 \lambda) a^{n+2}\]
  \begin{equation}\label{eq:game-rule} \sim \end{equation}
  \[(\sum_{i \in \N \setminus \{n, n+1, n+2\}} r_{i}a^{i}) + r_{n}a^{n} + (r_{n+1} + 3 \lambda)a^{n+1} + r_{n+2} a^{n+2}.\]
  To see this, assume we have a  sequence \((r_{i})_{i \in \mathbb{N}}\) with \(3 \lambda  + \sum_{i}r_{i} = 1\) and \(\lambda < \frac{1}{3}\) (if \(\lambda = \frac{1}{3}\) the statement is trivial).
  Let \(\chi = \frac{\lambda}{1 - 3 \lambda}\) (so \(\lambda = \frac{\chi}{3 \chi + 1}\)) and \(s_{i} = r_{i}(3 \chi + 1)\).
  Note that
  \[\sum_{i} s_{i} = (3 \chi + 1) \sum_{i} r_{i} = (3 \chi + 1)(1 - 3 \lambda) = 1, \]
  and \(\chi > 0\) since \(\lambda < \frac{1}{3}\).
  so \(\sum_{i}s_{i}a^{i}\) is an element of \(\mathcal{D} a^{*}\).
  We compute
  \begin{align*}
    & (\sum_{i \in \N \setminus \{n, n+1, n+2\}} r_{i}a^{i}) + (r_{n} + \lambda)a^{n} + r_{n+1}a^{n+1} + (r_{n+2} + 2 \lambda) a^{n+2} \\
    =\,\,& (\sum_{i \in \N \setminus \{n, n+1, n+2\}} \frac{s_{i}}{3 \chi + 1}a^{i}) + \frac{s_{n} + \chi}{3 \chi + 1}a^{n} + \frac{s_{n+1}}{3 \chi + 1}a^{n+1} + \frac{s_{n+2} + 2 \chi}{3 \chi + 1} a^{n+2} \\
    =\,\,& (\sum_{i \in \N \setminus \{n, n+1, n+2\}} \frac{s_{i}}{3 \chi + 1}a^{i}) + \frac{s_{n}}{3 \chi + 1}a^{n} + \frac{s_{n+1}}{3 \chi + 1}a^{n+1} + \frac{s_{n+2}}{3 \chi + 1}a^{n+2} \\
    &\quad+ \frac{3\chi}{3 \chi + 1}\frac{1}{3}a^{n}    + \frac{3 \chi}{3 \chi + 1}\frac{2}{3} a^{n+2} \\
    =\,\,& \frac{1}{3 \chi + 1}(\sum_{i \in \mathbb{N}} s_{i}a^{i}) + (1 - \frac{1}{3 \chi + 1}) (\frac{1}{3}a^{n}    +  \frac{2}{3}a^{n+2}) \\
    =\,\,& \frac{1}{3 \chi + 1}(\sum_{i \in \mathbb{N}} s_{i}a^{i}) + (1 - \frac{1}{3 \chi + 1}) (a^{n} +_{\frac{1}{3}} a^{n+2}) \\
    \sim\,\,& \frac{1}{3 \chi + 1}(\sum_{i \in \mathbb{N}} s_{i}a^{i}) + (1 - \frac{1}{3 \chi + 1}) a^{n+1} \\
    =\,\,& (\sum_{i \in \mathbb{N} \setminus \{n, n+1, n+1\}} \frac{s_{i}}{3 \chi + 1}a^{i}) + \frac{s_{n}}{3 \chi + 1} a^{n} + \frac{ s_{n+1} + 3\chi }{3 \chi + 1} a^{n+1} + \frac{s_{n+2}}{3 \chi + 1}a^{n+2} \\
    =\,\, & (\sum_{i \in \N \setminus \{n, n+1, n+2\}} r_{i}a^{i}) + r_{n}a^{n} + (r_{n+1} + 3 \lambda)a^{n+1} + r_{n+2} a^{n+2}. \\
  \end{align*}
This proves \eqref{eq:game-rule}. We now consider the following \emph{one-player game}:
  \begin{game}
    A game position is given by a distribution \(\varphi = \sum_{n} r_{n} \cdot a^{n} \in \mathcal{D} a^{*}\), denoted by an indexed row vector
    \[
      \begin{bmatrix}
        \cdots & n & n+1 & n+2 & \cdots \\
        \cdots & r_{n} & r_{n+1} & r_{n+2} & \cdots
      \end{bmatrix}.
    \]
    The game has only one bidirectional rule corresponding to \eqref{eq:game-rule}:
    given positions \[\varphi = (r_{n+1} + 3\lambda)a^{n+1} + \sum_{i \ne n+1} r_{i}a^{i} \text{ and } \varphi' = (r_{n} + \lambda)a^{n} + (r_{n+2} + 2 \lambda)a^{n+2} + \sum_{i \ne n, n+2} r_{i}a^{i}  \]
    for some \(\lambda\) with \(3\lambda + \sum_{i} r_{i} = 1\), the player can move between \(\varphi\) and \(\varphi'\) as indicated by
    {
      \begin{equation}
        \tag{RULE}
        \label{eq:rule}
        \begin{bmatrix}
          \cdots & n & n+1 & n+2 & \cdots  \\
          \cdots & r_{n} & \rn{11}{r_{n+1} + 3 \lambda} & r_{n+2} & \cdots \\
          \cdots & \rn{20}{r_{n} + \lambda} & r_{n+1} & \rn{22}{r_{n+2} + 2 \lambda} & \cdots
        \end{bmatrix}.
      \end{equation}
      \drawrule{11}{20}{11}{22}
    }

    We say that the rule has been applied at index \(n\).
    The \emph{winning positions} are the following positions for \(r_{n} \in [0, 1)\), corresponding to canonical distributions:
    \[
      \begin{bmatrix}
        \cdots & n - 1 & n & n+1 & n+2 & \cdots  \\
        \cdots & 0 & r_{n} & 1 - r_{n} & 0 & \cdots
      \end{bmatrix}.
    \]
  \end{game}

  By construction, we have \(\varphi \sim \varphi_{E(\varphi)}\) iff the player can reach a (necessarily unique) winning position when starting in the position corresponding to  \(\varphi\).
  So we have to prove that the player has a winning strategy for every starting position.
  We introduce some auxiliary terminology:
  \begin{defn}
    A \emph{hole} in a position \(\varphi \in \mathcal{D} a^{*}\) is given by a pair \([n, k] \in \N\) with \(k \ge 2\) such that \(\varphi(n) \ne 0, \varphi(n + k) \ne 0\) and \(\forall 1 \le j < k \colon \varphi(n + j) = 0\).
    The \emph{range} of a position \(\varphi \in \mathcal{D} a^{*}\) is defined as \[\rng \varphi = (\max_{\varphi(n) \ne 0} n) - (\min_{\varphi(n) \ne 0} n) + 1.\]
  \end{defn}
  Note that a position is winning iff it has range 1 or 2.
  Our strategy rests on the following two lemmas.
  \begin{lemma}[Spreading]\label{lem:spreading}
    From every position the player can reach a position without holes.
  \end{lemma}
\begin{proof}
  Let \(\varphi \in \mathcal{D} a^{*}\) be the current position.
  Let \([n,k]\) be a hole in \(\varphi\), we show how to ``patch'' it, i.e. how to get to a position that has one less hole.
  This means that \(\varphi\) looks locally like the first position in
  {
    \[
      \begin{bmatrix}
        \cdots & n - 1 & n     & n+1 & \cdots & n+k - 2& n+k-1 &  n+k & n+k+1 & \cdots  \\
        \cdots & r_{n-1}     & r_{n} & 0   & \cdots & 0 & 0 & \rn{11}{r_{n+k}} & r_{n+k+1} & \cdots \\
        \cdots & r_{n-1}     & r_{n} & 0   & \cdots & 0 & \rn{20}{\frac{r_{n+k}}{4}} & \frac{r_{n+k}}{4} & \rn{22}{r_{n+k+1}} + \frac{r_{n+k}}{2}& \cdots
      \end{bmatrix}.
    \]
    \drawrule{11}{20}{11}{22}
  }
  We apply \eqref{eq:rule} at \(n+k-1\) to split the coefficient at \(n+k\) with \(\lambda = \frac{r_{n+k}}{4}\), which makes the hole smaller (i.e.~decreases \(k\)) but does not create any new holes.
  If \(k = 2\) the hole is fixed, otherwise we repeat this process with new, smaller hole \([n, k-1]\) iteratively.
\end{proof}

  \begin{lemma}[Sweeping]\label{lem:sweeping}
    From every non-winning position without holes the player can reach a position without holes that has a strictly smaller range.
  \end{lemma}

\begin{proof}
  Let \(\varphi \in \mathcal{D} a^{*}\) be a non-winning position without holes.
  Let the minimal and maximal indices with non-zero coefficients be \(n \) and \(n+k\), respectively, so that \(\rng \varphi = k+1 > 2\) as \(\varphi\) is non-winning.
  Set \(\lambda = \min(\{ r_{n+i} \mid i < k\} \cup \{\frac{r_{n+k}}{2}\})\).
  We now ``sweep'' from left to right, applying \eqref{eq:rule} in sequence from right to left, starting with index \(n+k-2\) and stopping at index \(n\).
  {
    \[
      \begin{bmatrix}
        n     & n+1    & n+2    & \cdots &  n+k-3    & n+k-2     & n+k-1    &  n+k     \\
        r_{n} & r_{n+1} & r_{n+2} & \cdots & r_{n+k-3} & \rn{1lu}{r_{n+k-2}} & {r_{n+k-1}} & \rn{1ru}{r_{n+k}} \\
        r_{n} & r_{n+1} & r_{n+2} & \cdots & \rn{2lu}{r_{n+k-3}} & r_{n+k-2} - \lambda & \rn{1ld}{r_{n+k-1}} + \rn{1rd}{3 \lambda} & {r_{n+k} - 2 \lambda} \\
        r_{n} & r_{n+1} & r_{n+2} & \cdots & r_{n+k-3}\!-\!\lambda & \rn{2ld}{r_{n+k-2}} + \rn{2rd}{2 \lambda} & r_{n+k-1} +  \lambda & r_{n+k} - 2 \lambda \\
        \vdots & \rn{4lu}{\vdots} & \vdots & \rn{4ru}{\vdots} & \rn{3rd}{\vdots } & \vdots & \vdots & \vdots \\
        \rn{5lu}{r_{n}} & r_{n+1} \!-\!\lambda & \rn{4ld}{r_{n+2}}\! +\! \rn{4rd}{2 \lambda} & \cdots & r_{n+k-3} & r_{n+k-2} & r_{n+k-1} +  \lambda & r_{n+k} - 2 \lambda \\
        r_{n}\! -\! \lambda & \rn{5d}{r_{n+1}\! +\! 2 \lambda} & r_{n+2} & \cdots & r_{n+k-3} & r_{n+k-2} & r_{n+k-1} +  \lambda & r_{n+k} - 2 \lambda \\
      \end{bmatrix}
    \]
    \begin{tikzpicture}[overlay,remember picture]
      \draw [-, dotted, shorten >=4pt, shorten <=4pt] (1lu) -- (1ld);
      \draw [-, dotted, shorten >=4pt, shorten <=4pt] (1ru) -- (1rd);
      \draw [-, dotted, shorten >=4pt, shorten <=4pt] (2lu) -- (2ld);
      \draw [-, dotted, shorten >=4pt, shorten <=4pt] (1ld) -- (2rd);
      \draw [-, dotted, shorten >=5pt, shorten <=4pt] (2ld) -- (3rd);
      \draw [-, dotted, shorten >=7pt, shorten <=6pt] (4lu) -- (4ld);
      \draw [-, dotted, shorten >=4pt, shorten <=4pt] (4ru) -- (4rd);
      \draw [-, dotted, shorten >=4pt, shorten <=4pt] (5lu) -- (5d);
      \draw [-, dotted, shorten >=4pt, shorten <=4pt] (4ld) -- (5d);
    \end{tikzpicture}
  }
  We observe (1) that this constitutes a valid sequence of applications of \eqref{eq:rule}, i.e.~it does not happen that any index is \(< 0\) by definition of \(\lambda\);
  (2) after one such sweep all fields with index in \(\{n+1, \ldots, n+k-1\}\) did not decrease in value;
  (3) the fields at \(n\) and \(n+k\) strictly decreased in value.
  We continue sweeping \(j-1\) times with the same \(\lambda\) until \(r_{n} - j \lambda \le \lambda\) or \(r_{n+k} - 2 \lambda \le 2j \lambda\).
  Note that (2) above ensures that we can indeed continue with the same \(\lambda\).
  The new position now looks like this:
  \[
    \begin{bmatrix}
      n     & n+1    & n+2    & \cdots &  n+k-3    & n+k-2     & n+k-1    &  n+k    \\
      r_{n} - j \lambda & r_{n+1} + 2j \lambda & r_{n+2} & \cdots & r_{n+k-3} & r_{n+k-2} & r_{n+k-1} +  j\lambda & r_{n+k} - 2 j \lambda \\
    \end{bmatrix}.
  \]
  If now \(r_{n} - j \lambda = 0\) or \(r_{n+k} - 2j \lambda = 0\) we are done: the new position has a strictly smaller range.
  Otherwise, we perform one last sweep with \(\lambda' = \min(r_{n} - j \lambda,\, r_{n+k}/2 - j \lambda) \le \lambda\).
  Afterwards either the field at \(n\) or at \(n+k\) (or both) have value \(0\), so this position has a range smaller than \(\varphi\), and we are done.
\end{proof}

  Now given a starting position \(\varphi \in \mathcal{D} a^{*}\) we first use Spreading to reach a position without holes. We then repeatedly apply Sweeping until we reach a position that has range at most 2. This is a winning position. \qedhere
\end{proof}

\appendixproof{thm:em-recognition}
We start with a technical remark:
\begin{rem}\label{rem:composition}
\begin{enumerate}[(1)]
\item Given a $\T$-algebra $A$, the map $(X\kleislito Y)\xto{-\seq g^\#} (X\to A)$ is a $\T$-morphism for every $g\colon Y\to A$, as it is just the product map $(\T Y)^X \xto{(g^\#)^X} A^X$. 
\item The map $(Y\kleislito Z)\xto{f \kseq -} (X\kleislito Z)$ is a $\T$-morphism for every \emph{pure} map $f\colon X\to Y$. (For commutative $\T$, this also holds for non-pure $f\colon X\kleislito Y$.)
\end{enumerate}
\end{rem}

Now we turn to the proof of the theorem.

\begin{proof}
(1)$\Rightarrow$(2) Let $\A=(Q,\delta,i,o)$ be a $\T$-FA computing $L$. We may assume that the initial state $i$ is pure (\Cref{rem:t-aut-props}). Then $L$ is recognized by the \finite $(\T,\M)$-bialgebra $Q\kleislito Q$ via the monoid morphism $\bar\delta^*\colon\Sigma^*\to (Q\kleislito Q)$ and the predicate $p\colon (Q\kleislito Q) \to O$ given by $p(f)=(i\kseq f \seq o^\#)$. Note that $p$ is a $\T$-morphism by \Cref{rem:composition}. Indeed, $L=\bar\delta^*\seq p$ holds because for all $w\in \Sigma^*$,
\[ L(w)=(i\kseq \bar\delta^*(w)\seq o^\#) = p(\bar\delta^*(w)).  \]
(2)$\Rightarrow$(1) Let $M$ be an \finite $(\T,\M)$-bialgebra recognizing $L$ via a monoid morphism $h\colon \Sigma^*\to M$ and a $\T$-morphism $p\colon M\to O$. Since $M$ is finitely generated as a $\T$-algebra, there exists a surjective $\T$-morphism $q\colon T G\epito M$ for some finite set $G$. Let $(M,\mult,e)$ denote the monoid structure and $q_0\colon G\to M$ and $h_0\colon \Sigma\to M$ the respective domain restrictions of $q$ and $h$. We construct a $\T$-FA $\A=(G,\delta,i,o)$ computing $L$ as follows.  Choose $i\colon 1\kleislito G$ and $\delta\colon G\times \Sigma \kleislito G$ such that the first two diagrams below commute; such choices exist because $q$ is surjective. Moreover, define $o\colon G\to O$ by $o:=q_0\seq p$ as in the third diagram.
\[
\begin{tikzcd}
1 \ar{r}{i} \ar{dr}[swap]{e} & TG \ar[two heads]{d}{q} \\
& M
\end{tikzcd}
\qquad
\begin{tikzcd}
G\times \Sigma \ar{rr}{\delta} \ar{d}[swap]{q_0\times \id} && TG \ar[two heads]{d}{q} \\
M\times \Sigma \ar{r}{\id\times h_0} & M\times M \ar{r}{\mult} & M
\end{tikzcd} 
\qquad
\begin{tikzcd}
G \ar{d}[swap]{q_0} \ar{r}{o} & O \\
M \ar{ur}[swap]{p} & 
\end{tikzcd}
\]
To show that $\A$ computes the language $L$, fix $w=a_1\cdots a_n\in \Sigma^*$. Then the two diagrams below (in $\Kl(\T)$ and $\Set$, respectively) commute:
\[
\begin{tikzcd}[column sep=30, row sep=5]
& G \ar[kleisli]{r}{\bar\delta(a_1)} \ar{dd}{q_0} & G \ar{dd}{q_0} \ar[phantom]{r}{\cdots} & G \ar[kleisli]{r}{\bar\delta(a_n) } \ar{dd}{q_0} & G \ar{dd}{q_0} \\
1 \ar[kleisli]{ur}{i} \ar{dr}[swap]{e} & & & &  \\
& M \ar{r}{-\mult h_0(a_1)} & M \ar[phantom]{r}{\cdots} & M \ar{r}{-\mult h_0(a_n)} & M
\end{tikzcd}
\quad\quad
\begin{tikzcd}[column sep=30, row sep=5]
 TG \ar{dd}[description]{q=q_0^\#} \ar{dr}{o^\#} & \\
 & O \\
M \ar{ur}[swap]{p}& 
\end{tikzcd}
\]
Since $h(w)=e\mult h_0(a_1)\mult \cdots h_0(a_n)$, we see that
\[L(w)=(h\seq p)(w)= (i\kseq \bar\delta^*(w)\seq o^\#).\]
This proves that the automaton $\A$ computes the language $L$, as claimed.

\medskip\noindent
Now suppose that the monad $\T$ is commutative. Then (1)$\Rightarrow$(3) is shown like (1)$\Rightarrow$(2), together with the observation that the recognizing $(\T,\M)$-bialgebra $Q\kleislito Q$ is a $\T$-monoid (\Cref{ex:fg-carried}\ref{ex:x-kto-x}). The implication (3)$\Rightarrow$(2) is trivial. 
\end{proof}

\detailsfor{R:prod-fg}

The condition of \Cref{thm:em-recognition}, which states that for every finite set $X$ the
\T-algebra $(TX)^{X}$ is finitely generated, is equivalent to the simpler condition that
finitely generated \T-algebras are closed under finite products.

        This simpler condition trivially implies the original one. To show that it is equivalent, suppose $A, B$ are finitely generated \T-algebras, we have to show that their product is finitely generated.
        If one of $A, B$ is empty, then so is their product $A \times B$, which is thus also finitely generated.
        Otherwise, there exist surjective \T-algebra morphisms $e_{A} \colon TX \epi A$ and $e_{B} \colon TY \epi B$ for finite, non-empty sets $X, Y$.
        Choose a set $Z$ of cardinality larger than those of $X, Y$, which contains at least two elements, and surjections $f \colon Z \epi X, g \colon Z \epi Y$. By assumption the \T-algebra $(TZ)^{Z}$ is finitely generated, so there exists a finite set $K$ and a  surjection $h \colon TK \epi (TZ)^{Z}$. Since every endofunctor on \Set preserves surjections the following map is a surjective \T-algebra morphism, where the morphism $p$ is the projection $(TZ)^{Z} \cong (TZ)^{2} \times TZ^{|Z|-2}$ on the first two components -- it exists since $|Z| \ge 2$.
        \[
        \begin{tikzcd}[column sep = 35]
          TK \ar[two heads]{r}{h} & (TZ)^{Z} \ar[two heads]{r}{p} & (TZ)^{2} \cong TZ \times TZ \ar[two heads]{r}{Tf \times Tg} & TX \times TY \ar[two heads]{r}{e_{A} \times e_{B}} & A \times B.
        \end{tikzcd}
        \]
        This shows that the \T-algebra $A \times B$ is indeed finitely generated.

\detailsfor{ex:xx-fg}

See `Details for \Cref{ex:xx-fg,ex:xx-mon-fg}' further below. 

\appendixproof{thm:fg-carried-implies-fp-general}
The statement is immediate from \Cref{cor:fb-impl-fp} and the assumption that finitely generated and finitely presentable $\T$-algebras coincide. \qed

\detailsfor{rem:mon-functor}
It is folklore that if $\T$ is commutative, then, for every $\T$-monoid $N$, the extension $h^{\#} \colon TM \rightarrow N$
of every monoid morphism $h \colon M \rightarrow N$ is also a monoid morphism.
This can be seen as follows: the condition that $N$  is a $\T$-monoid can be expressed equivalently in purely diagrammatical terms by requiring commutativity of Diagram \ref{eq:t-monoid}.
\begin{equation}
  \label{eq:t-monoid}
  \begin{tikzcd}
    TN \times TN \rar{\pi_{N, N}} \dar{ n \times n} & T(N \times N) \rar{T(\mult)} & TN \dar{n} \\
    N \times N \ar{rr}{\mult} &&  N
  \end{tikzcd}
\end{equation}
In the terminology of~\cite{seal-2012}, the multiplication of $N$ is a \emph{$\T$-bimorphism}. [Note that this condition makes sense for any -- not necessarily non-commutative -- double-strength \(\pi_{X, Y}\).]

The general statement is as follows: Let \((N, n, \mult)\) be a \((\T, \M)\)-bialgebra satisfying Diagram~\ref{eq:t-monoid}. Then the extension \(f^{\#} \colon TM \rightarrow N\) of every monoid morphism $f \colon M \rightarrow N$ into \((N, \mult)\) is also a monoid morphism.
Here $TM$ is equipped with the obvious multiplication \(TM \times TM \rightarrow T(M \times M) \rightarrow TM\) induced by the double strength, and in general it does not satisfy Diagram~\ref{eq:t-monoid} in general.
The fact that \(f^{\#}\) is a monoid morphism is witnessed by the following commutative diagram:
\begin{equation}
  \begin{tikzcd}
    TM \times TM \rar{\pi_{M, M}} \dar{Tf \times Tf} \ar[swap,shiftarr={xshift=-30pt}]{dd}{f^{\#} \times f^{\#}} & T(M \times M) \rar{T \mult} \dar{T(f \times f)} & TM \dar{T f} \ar[shiftarr={xshift=20pt}]{dd}{f^{\#}} \\
    TN \times TN \rar{\pi_{N, N}} \dar{ n \times n} & T(N \times N) \rar{T(\mult)} & TN \dar{n} \\
    N \times N \ar{rr}{\mult} &&  N
  \end{tikzcd}
\end{equation}
The top left and right squares commute by naturality of \(\pi\) and since $f$ is a monoid morphism, respectively, and the bottom rectangle is just Diagram~\ref{eq:t-monoid} again.

\appendixproof{thm:kleisli-recognition}

\begin{rem}
In the proof we use that the free monoid $\Sigma^*$ is \emph{projective}: Given a surjective monoid morphism $e\colon M\epito N$ and a morphism $g\colon \Sigma^*\to N$, there exists a morphism $h\colon \Sigma^*\to M$ such that $g=h\seq e$. Indeed, we can choose $h$ to be the free extension of $g_0\seq s$, where $g_0\colon \Sigma\to N$ is the domain restriction of $h$ and $s\colon N\to M$ is a splitting of the surjection $e$, that is, a map such that $s\seq e=\id_N$. 
\end{rem}

\begin{proof}
($\Rightarrow$) Suppose that $\A$ is a $\T$-FA with states $Q$ computing $L$. From the proof of \Cref{thm:em-recognition} we know that $L$ is recognized by the $(\T,\M)$-bialgebra $Q\kleislito Q$; say $L=g\seq p$ for some monoid morphism $g\colon \Sigma^*\to (Q\kleislito Q)$ and some $\T$-morphism $p\colon (Q\kleislito Q)\to O$. (The specific choices of $g$ and $p$ in that proof are not relevant for the present argument.) Since $Q\kleislito Q$ is monoidally finitely generated, there exists a monoid morphism $e\colon M\to (Q\kleislito Q)$ such that $M$ is a finite monoid and $e^\#\colon TM\epito (Q\kleislito Q)$ is a surjective monoid morphism. By projectivity of $\Sigma^*$, there exists a monoid morphism $h\colon \Sigma^*\to TM$ such that $g=h\seq e^\#$. Putting $q := e\seq p$, we thus have the following commutative diagram:
\[
\begin{tikzcd}
\Sigma^* \ar{dr}[swap]{g} \ar{r}{h} & TM \ar[two heads]{d}{e^\#} \ar{r}{q^\#} & O \\
& Q\kleislito Q \ar{ur}[swap]{p} & 
\end{tikzcd}
\] 
Since $h$ is monoid morphism, the diagram below in $\Set$ commutes.
\begin{equation}\label{eq:tm-diag}
\begin{tikzcd}[column sep=20]
\Sigma^*\times \Sigma^* \ar{rr}{\mult} \ar{d}[swap]{h\times h} && \Sigma^* \ar{d}{h} & \ar{l}[swap]{\varepsilon} 1 \ar{d}{e} \\
TM\times TM \ar{r}{\pi_{M,M}} & T(M\times M) \ar{r}{T\mult} & TM & M \ar{l}[swap]{\eta_M}
\end{tikzcd}\end{equation}
This is equivalent to commutativity of the following diagram in $\Kl(\T)$, which states that $h\colon \Sigma^*\kleislito M$ is an effectful monoid morphism:
\begin{equation}\label{eq:tm-diag2}
\begin{tikzcd}[column sep=15]
  \Sigma^*\times \Sigma^*
  \arrow{r}{\mult}
  \arrow[kleisli]{d}[swap]{h\bar\times h}
  & \Sigma^*
  \arrow[kleisli]{d}[swap]{h}
  & 1
  \arrow{l}[swap]{\varepsilon}
  \arrow{dl}{e}
  \\
  M \times M
  \arrow{r}{\mult}
  &
  M
\end{tikzcd}
 \end{equation}
The finite monoid $M$ recognizes $L$ via $h$ and $q$ because $L=g\seq p=h\seq q^\#$.

\medskip\noindent 
($\Longleftarrow$) Suppose that the finite monoid $(M,\mult,e)$ effectfully recognizes the language $L$ via the effectful monoid morphism $h\colon \Sigma^*\kleislito M$ and the map $p\colon M\to O$. Then \eqref{eq:tm-diag2} or equivalently \eqref{eq:tm-diag} commutes, so $h$ is a monoid morphism from $\Sigma^*$ to the monoid $TM$. Thus $\T M$ is an \finite $(\T,\M)$-bialgebra recognizing $L$ via $h$ and $p^\#$. From \Cref{thm:em-recognition} we conclude that there exists a $\T$-FA computing $L$. Specifically, by instantiating the proof of (2)$\Rightarrow$(1) to the $(\T,\M)$-bialgebra $\T M$ and the choices $G=M$ and $q=\id_{TM}$, we see that a $\T$-FA computing $L$ is given by $\A=(M,\delta,e,p)$ with the transitions
\[ \begin{tikzcd} (M\times \Sigma \ar[kleisli]{r}{\delta} & M)\end{tikzcd} = \begin{tikzcd} (M\times \Sigma \xto{\eta_M\times h_0} TM\times TM \xto{\pi_{M,M}} T(M\times M) \xto{T\mult} TM).\end{tikzcd}\]
\end{proof}

\appendixproof[Details for \autoref{ex:xx-fg} and]{ex:xx-mon-fg}
We verify that for $\T\in \{\D,\C,\S\}$ and each finite set $X$ the maps $\xi_{0}\colon M\to (X\kleislito X)$ of \Cref{fig:witnesses} freely extend to $\T$-morphisms $\xi = \xi_{0}^{\#}\colon M\epito (X\kleislito X)$ that are both surjective and monoid morphisms, and thus witness that the $\T$-bialgebra $X\kleislito X$ is monoidally finitely generated. For the non-commutative monads $\C$ and $\S$, we show that $\xi$ is a monoid morphism using \Cref{prop:central-implies-mon-morphism}.

\subparagraph{Distribution monad $\D$:}
Since $\D$ is an affine monad ($\D1\cong 1$), the affine map
\[
  \xi\colon \D(X\to X)\to (X\to \D X)
\]
extending the map $\xi_{0}\colon (X\to X)\to (X\to \D X)$ given by $\xi_{0}(f)=f\seq \eta_X$ is a surjective monoid morphism by~\Cref{rem:affine-implies-xx-mon-fg}.

\subparagraph{$S$-semimodule monad $\S$:}
Analogous to probability distributions, we may represent an element $f$ of $\S X$ as a formal finite sum $\sum_{i\in I} s_i x_i$ where $x_i\in X$, $s_i\in S$ and $f(x)=\sum_{i\in I:\, s_i=x} s_i$ for all $x\in X$. Consider the map $\xi_{0}\colon (X\parfun X)\to (X\to \S X)$ as in \Cref{fig:witnesses}.
\begin{enumerate}[(1)]
\item $\xi$ is surjective: Let $n=|X|$. Then the $(\S,\M)$-bialgebra $n\kleislito n$ is, up to isomorphism, the set of $n\times n$-matrices over the semiring $S$ with the usual multiplication, addition and scalar multiplication of matrices. Moreover $n\parfun n$ is the submonoid of matrices with all entries in $\{0,1\}$ where each row contains at most one $1$ entry. Then $\xi_{0}$ is just the inclusion, and $\xi$ is given by $\sum_i s_i A_i \mapsto \sum_i s_i A_i$, i.e.~a formal linear combination of matrices in $n\parfun n$ is mapped to the actual linear combination in the semimodule $n\kleislito n$.
Since each $n\times n$-matrix over $S$ can be expressed as a linear combination of matrices in $n\parfun n$ (in fact, matrices with a single entry $1$, and $0$ otherwise), we see that $\xi$ is surjective.
\item $\xi$ is a monoid morphism: By \Cref{prop:central-implies-mon-morphism} we only need to show that the uncurried form $f\colon X\times (X\parfun X)\to \S X$ of $\xi_{0}$ is central. The map $f$ is given by $f(x,k)=1\mult k(x)\in \S X$ if $k(x)$ is defined, and $0\in \S X$ otherwise.
 Instantiating \eqref{eq:central} to $f$ and $f'\colon X'\to \S Y'$, wee see that the two maps $X\times (X\parfun X)\times X'\to \S(X\times Y')$ of \eqref{eq:central} both send $(x,k,x')$ to $\sum_i s_i (k(x),y_i')$ (where $f_i'(x')=\sum_i s_i y_i'$) if $k(x)$ is defined, and to $0\in \S(X\times Y')$ otherwise. Thus, $f$ is central.
\end{enumerate}

\subparagraph{Convex powerset of distributions monad $\C$:}
The free extension of $\xi_{0} \colon (X\to X)\to (X\to \C X)$ as in \Cref{fig:witnesses} is given by
 \[\xi \colon \Cp(X \rightarrow X) \rightarrow (X \rightarrow \Cp X), \qquad \xi(U)(x) = \{ \sum_{i} r_{i} \cdot f_{i}(x) \mid \sum_{i} r_{i} \cdot f_{i} \in U\}. \]
\begin{enumerate}[(1)]
\item $\xi$ is surjective: We give a categorical proof using affinity of the non-empty finite power set monad \(\mathcal{P}_{\mathrm{f}}^{+}\) and the finite distribution monad \(\mathcal{D}\). For a convex set \(A\) the set of its finitely generated non-empty convex subsets is itself a convex set \(C A\), and this yields a monad on $\Conv\cong \Alg(\D)$.
    Then  \(\Cp = \D \seq   C \seq |-| \) is the composite of the monad \(C \) with the adjunction \(\mathcal{D} \dashv |-|\), where $|-|\colon \Conv\to \Set$ denotes the forgetful functor.
    In \cite[Lemma 2]{hhos18} it is shown that taking convex hulls is a natural transformation \(\langle-\rangle_{A} \colon |\Pfp A| \rightarrow | \Cp A| \).
    [Note that it is not natural for the endofunctors \(\Pfp, \Cp\) on \conv since \(\langle - \rangle \) is not affine.]
    We then get the following diagram in \Set, where $\xi^{\Pfp}$ and  $\xi^\D$ are the free extensions of the respective inclusion maps $(X\to X)\to (X\kleislito X)$.
    \[
      \begin{tikzcd}
        \Pfp \mathcal{D} (X \rightarrow X) \rar[two heads]{\Pfp \xi^{\mathcal{D}}} \dar[two heads]{\langle - \rangle}
        &
        \Pfp (X \rightarrow \mathcal{D} X) \rar[two heads]{\xi^{\Pfp}} \dar[two heads]{\langle - \rangle}
        &
        (X \rightarrow \Pfp \mathcal{D} X) \dar[two heads]{ \langle - \rangle }
        \\
        C  \mathcal{D} (X \rightarrow X) \rar{C  \xi^{\mathcal{D}}} \dar[equal]{}
        &
        C (X \rightarrow \mathcal{D} X) \rar{\langle C  \ev_{x} \rangle_{x \in X}}
        &
        (X \rightarrow C \mathcal{D} X) \dar[equal]{}
        \\
        \Cp(X \rightarrow X) \ar{rr}{\xiCp}
        &
        &
        (X \rightarrow \Cp X)
      \end{tikzcd}
    \]
    The two upper squares commute by naturality of $\langle - \rangle$ and the bottom rectangle commutes since it does if we postcompose with the projections \(\ev_{x}\):
    \[C  \xi^{\mathcal{D}}\seq \langle C  \ev_{x} \rangle_{x \in X} \seq \ev_{x} = C  \xi^{\mathcal{D}}  \seq C  \ev_{x} = C ( \xi^{\mathcal{D}} \seq \ev_{x}) = C  \mathcal{D} \ev_{x} = \Cp \ev_{x} = \xi \seq \ev_{x} .\]
    The composition \( \langle - \rangle \seq \xi  \) is surjective since the top right path is surjective by affinity of \(\Pfp\) and \(\mathcal{D}\), see~\Cref{rem:affine-implies-xx-mon-fg}). In particular, \(\xi\) is surjective.
\item $\xi$ is a monoid morphism:  By \Cref{prop:central-implies-mon-morphism} we only need to show that the uncurried form $f\colon X\times (X\to X)\to \C X$ of $h$ is central. The map $f$ is given by $f(x,k)=\{1\cdot k(x)\}\in \C X$.  Instantiating \eqref{eq:central} to $f$ and $f'\colon X'\to \C Y'$, wee see that the two maps $X\times (X\to  X)\times X'\to \C(X\times Y')$ of \eqref{eq:central} both send $(x,k,x')$ to \[\{ \sum_{i\in I_j} s_i^j (k(x),y_i^j) \mid j\in J \} \in \C(X\times Y') \qquad \text{where } f'(x')= \{ \sum_{i\in I_j} s_i^j y_i^j \mid j\in J \}.\]
\end{enumerate}

\subsection*{A Separating Example for \Cref{thm:em-recognition,thm:kleisli-recognition}}%
\addcontentsline{toc}{subsection}{A separating counterexample for \Cref{thm:em-recognition,thm:kleisli-recognition}}
\label{app:sep-ex}

While all monads considered in this paper satisfy both \Cref{thm:em-recognition,thm:kleisli-recognition} (even the list monad \(\mathcal{M}\)),
we show here that the condition of \Cref{thm:em-recognition} is indeed strictly stronger than that of \Cref{thm:kleisli-recognition} by giving an example of a monad \T{} satisfying the conditions of the former but not of the latter.
Consider the monad \T{} corresponding to the algebraic variety $(\Lambda, E)$ given by operations \(\Lambda = \{+/2, e/1, \alpha/1\}\) with equations
\begin{equation}
  \label{eq:E}
  E = \{x + e = x, e + x = x\}.
\end{equation}

\begin{lemma}
  The monad \T{} satisfies the condition of \Cref{thm:em-recognition}.
\end{lemma}
\begin{proof}
  We have to show that for every finite set $X$ the \T-algebra $X \rightarrow TX$ is finitely generated.
  Since $(X \kleislito X) \cong (TX)^{X}$ is a finite product of finitely generated free  $\T$-algebras, it suffices to show that finitely generated $\T$-algebras are closed under finite products. Given finitely generated \T-algebras $A, B$ with generating sets $A_{0} \subseteq A, B_{0} \subseteq B$, it is easy to see that $A \times B$ is generated by the set $A_{0} \times \{e\} \cup \{e\} \times B_{0}$, since $(a, b) = (a + e, e + b) = (a, e) + (e, b)$.
\end{proof}

\begin{lemma}
  The monad \T{} does not satisfy the condition of \Cref{thm:kleisli-recognition}.
\end{lemma}
Note in particular that \Cref{prop:central-implies-mon-morphism} is not satisfied, the conceptual reason being that all central morphisms in $\Kl(\T)$ are pure.
\begin{proof}
  We show that for $X = 2$ the $(\T,\M)$-bialgebra $X \to T X$ is not monoidally finitely generated.
  Concretely, this means that for every finite monoid $M$ and a monoid morphism $e \colon M \rightarrow (2\to T 2)$ the extension $e^{\#} \colon TM \epi (2 \to T 2)$ cannot be both surjective and a monoid morphism
  Note that we may assume, without loss of generality, that $M$ is a submonoid of $2 \to T2$ and $e = \iota$ is its inclusion, since we can always factorize $e$ through its image $\iota \colon \mathrm{Im}(e) \hookrightarrow (2 \to T2)$.
  So we have to show that for every submonoid $\iota \colon M \incl (2 \rightarrow T 2)$, if $\iota^{\#}$ is surjective the following diagram does not commute:
  \begin{equation}
    \label{eq:mor-diag}
    \begin{tikzcd}
      TM \times TM \ar{r}{\pi_{M, M}} \ar[shiftarr={xshift=-60pt}, swap]{dd}{\iota^{\#} \times \iota^{\#}} \ar{d}{T \iota \times T \iota} & T(M \times M) \ar{r}{T \kseq} & TM \ar{d}{T \iota} \ar[shiftarr={xshift=30pt}]{dd}{\iota^{\#}} \\
      T(2 \to T 2) \times T(2 \to T 2) \ar{d}{\mu_{2} \times \mu_{2}} & & T(2 \to T2) \ar{d}{\mu_{2}} \\
      (2 \to T2) \times (2 \to T2) \ar{rr}{\kseq} && 2 \to T2
    \end{tikzcd}
  \end{equation}

  \medskip\noindent(1) First we show that a subset $\iota \colon M \incl (2 \rightarrow T 2)$ generating the \T-algebra $2 \rightarrow T 2$ must contain some non-pure map $f \colon 2 \to T 2$, viz.\ a map $f$ that is not of the form $f' \seq \eta_{2}$ for some $f' \colon 2 \rightarrow 2$. This is the case, since even the full subset $M = (2 \rightarrow 2) \incl (2 \rightarrow T2)$ of all pure maps does not generate the algebra $2 \to T 2$. The algebra morphism $\iota^{\#} \colon T(2 \rightarrow 2) \rightarrow (2 \rightarrow T 2)$ is not surjective, as it concretely is defined as
  \begin{align*}
    \iota^{\#} \colon T(2 \rightarrow 2) &\rightarrow (2 \rightarrow T 2) \\
    t(f_{0}, \ldots, f_{3}) &\mapsto (x \mapsto t(f_{0}(x), \ldots, f_{3}(x))),
  \end{align*}
  where $t(f_{0}, \ldots , f_{3})$ is a $\Lambda$-term in the variables $f_{i} \in (2 \rightarrow 2)$.
  It is easy to see that the Kleisli morphism sending $0 \mapsto e, 1 \mapsto \alpha$ cannot be in the image of $\iota^{\#}$, as it would require both $e = t $ and $t =  \alpha$.

  \medskip\noindent(2) We show that if $M$ contains a non-pure map $f \colon 2 \rightarrow T2$ then diagram~\ref{eq:mor-diag} does not commute.
  More specifically, we show for every non-pure Kleisli map $f \colon 2 \rightarrow T2$ there exists a term $v \in TM$ such that the pair $(\eta_{M}(f), v(f, f)) \in TM \times TM$ is sent by the upper right and lower left path of Diagram~\ref{eq:mor-diag} to different elements of $2 \rightarrow T 2$.
  Let us first simplify the notation.
  We identify a Kleisli map $f \colon 2 \to T 2 $ with its graph tuple $[s(0, 1), t(0, 1)] \in (T2)^2$, where $f(0) = s(0, 1) \in T 2, f(1)= t(0, 1) \in T 2$ are $\Lambda$-terms in $0, 1 \in 2$.
  The Kleisli composite $f \kseq f$ corresponds to the tuple $[s(s(0, 1), t(0, 1)), t(s(0, 1), t(0, 1))]$.
  Suppose that $v(f) \in TM$ is a $\Lambda$-term with variable $f \in M$.
  The pair $(\eta_{M}(f), v(f)) \in TM \times TM$ is sent by the upper right path of Diagram~\ref{eq:mor-diag} to $\iota^{\#}(\eta_{M}(f) \mult v(f)) \in (2 \rightarrow T 2)$, which is given by
  \begin{align}
   [v(s(s(0, 1), t(0, 1))),
v(t(s(0, 1), t(0, 1)))].\label{eq:ur}
\end{align}
By the lower left path of Diagram~\ref{eq:mor-diag} it is sent to $\iota^{\#}(\eta_{M}(f)) \mult \iota^{\#}(v(f, f))$, given by
\begin{align}
  [s(v(s(0, 1)), v(t(0, 1))), t(v(s(0, 1)), v(t(0, 1)))] \label{eq:dl}
\end{align}
We now give for every possible value of $f(0) \in T2$  a term $v(f) \in TM$ such $\iota^{\#}(\eta(f) \mult v(f)) \ne \iota^{\#}(\eta(f)) \mult \iota^{\#}(v(f))$.
If $f(0) = s(0, 1) = e$ then we set $v(f) = \alpha$ to get
\[\iota^{\#}(\eta(f) \mult v(f))(0) = \text{(\ref{eq:ur})}(0) = \alpha \ne e = \text{(\ref{eq:dl})}(0) =   \big(\iota^{\#}(\eta(f)) \mult \iota^{\#}(v(f))\big)(0).\]
Similarly if $f(0) = s(0, 1) = \alpha$ then we set $v(f) = e$ to get
\[\iota^{\#}(\eta(f) \mult v(f))(0) = \text{(\ref{eq:ur})}(0) = e \ne \alpha = \text{(\ref{eq:dl})}(0) =  \big(\iota^{\#}(\eta(f)) \mult \iota^{\#}(v(f))\big)(0).\]
Finally, if $f(0) = s(0, 1) = c(0, 1) + d(0, 1)$ for terms $c(0, 1), d(0, 1) \in T2$, then may assume by \Cref{eq:E} that $c(0, 1) \ne e$ and $d(0, 1) \ne e$, and we set $v(f) = \alpha$.
Suppose now for the sake of contradiction that we have an equality \[\iota^{\#}(\eta(f) \mult v(f))(0) = \text{(\ref{eq:ur})}(0) = \alpha  = c(\alpha, \alpha) + d(\alpha, \alpha) = \text{(\ref{eq:dl})}(0) =  \big(\iota^{\#}(\eta(f)) \mult \iota^{\#}(v(f))\big)(0). \]
The term $c(\alpha, \alpha) + d(\alpha, \alpha)$ is a ground $\Lambda$-term, i.e., it does not contain any variables from $M$.
It is clear from the equations of the theory that for all ground terms $x, y $ we have $x + y = \alpha$ if and only if $x = \alpha$ and $y = e$ or $x = e$ and $ y = \alpha$. But this means that either $c(\alpha, \alpha) = e$ or $d(\alpha, \alpha) = e$, and w.l.o.g. we assume the former. If $c(\alpha, \alpha) = e$ then $c$ can not contain any occurence of the variables $0, 1$. This implies that $c(0, 1) = e$, a contradiction to our assumptions $c(\alpha, \alpha) \ne e \ne d(\alpha, \alpha)$.

We conclude that there does not exist a submonoid $\iota  \colon M \incl (2 \to T 2)$ such that the extension is both surjective and a monoid morphism, so $2 \rightarrow T 2$ is not monoidally finitely generated. \qedhere
\end{proof}

\appendixproof{prop:central-implies-mon-morphism}

\begin{proof}
  \newcommand{\trr}{\vartriangleright}
  \newcommand{\trl}{\vartriangleleft}
\savebox{\kleisliarrow}{%
\begin{tikzpicture}[
      baseline=(arrow.base),
      inner sep=0mm,
      outer sep=0mm,
      ]
      \node[draw=none,
      anchor=base,
      overlay,
      inner sep=0,
      outer sep=0,
      minimum height=1em,
      ] (arrow) {$\phantom{\longrightarrow}$};
      \begin{scope}[even odd rule,overlay]
        \clip  ($ (arrow.south) !.68! (arrow.north)$) circle (0.17em)
           (arrow.north west) rectangle (arrow.south east);
      \node[draw=none,
      anchor=base,
      inner sep=0,
      outer sep=0,
      ] (arrow) {$\longrightarrow$};
    \end{scope}
    \draw[fill=none,overlay] ($ (arrow.south) !.68! (arrow.north)$) circle (0.15em);
    \draw[use as bounding box,draw=none] (arrow.north west) rectangle (arrow.south east);
  \end{tikzpicture}}

  The left and right strength of the monad $\T$ entail that $\Kl(\T)$ has the structure of a \emph{strict premonoidal category}:
  For ev
  ery set \(X\) there exist functors \(X \trl (-), (-) \trr X \colon \Kl(\T) \rightarrow \Kl(\T) \) with \(X \trl Y = X \times Y = X \trr Y\) that are defined on morphisms \(f \colon X \kleislito Y\) by
  \[X' \trl f = (\id_{X'} \times f) \kseq \lst_{X', Y} \colon X' \times X \rightarrow X' \times TY \kleislito X' \times Y  \]
  and
  \[f \trr X' = (f \times \id_{X'}) \kseq \rst_{X', Y} \colon X \times X' \rightarrow TY \times X' \kleislito Y \times X',  \]
 and the following equations hold:
  \begin{align*}
   \id_{X} \trr X' &= \rst_{X, X'},\\
    \qquad X' \trl \id_{X} &= \lst_{X', X},\\
   (X \trl \id_{Y}) \trr Z &= X \trl (\id_{Y} \trr Z).
\end{align*}
  Now it is equivalent to say that a morphism \(f \colon X \kleislito Y\) in $\Kl(\T)$ is central iff for all \(f' \colon X' \kleislito Y'\) the following square commutes:
  \[
    \begin{tikzcd}
      X \times X' \rar[kleisli]{f \trr X'} \dar[kleisli][swap]{X \trl f'}
      &
      Y \times X' \dar[kleisli]{Y \trl f'} \\
      X \times Y' \rar[kleisli]{f \trr Y'} & Y \times Y'
    \end{tikzcd}
  \]
  Now assume that \(\bar{f} \colon M \rightarrow (TX)^{X}\) is a monoid morphism such that \(f \colon X \times M \rightarrow TX\) is central and let \(\bar{f}^{\#} \colon TM \rightarrow (TX)^{X}\) be the extension of $f$ to a \(\T\)-morphism.
 Then the three squares below commute, where $\ev$ is the evaluation morphism and $m$ is the multiplication of $M$.
  \[
    \begin{tikzcd}[column sep=small]
      X \!\times\! TM \rar[kleisli]{X \trl \id} \dar[swap]{\id \!\times\! \bar{f}^{\#}}
      &
      X \!\times\! M \dar[kleisli]{f}
      &
      &
      X \!\times\! (TX)^{X} \!\times\! (TX)^{X} \rar{\id \!\times\! \kseq}
      \dar[kleisli][swap]{\ev \trr (TX)^{X}}
      &
      X \!\times\! (TX)^{X}
      \dar[kleisli]{\ev}
      \\
      X \!\times\! (TX)^{X} \rar[kleisli]{\ev}
      &
      X
      &
      &
      X \!\times\! (TX)^{X}
      \rar[kleisli]{\ev}
      &
      X
    \end{tikzcd}
  \]
\[
    \begin{tikzcd}[column sep=small]
      X \!\times\! M \!\times\! M \rar{\id \!\times\! m} \dar[kleisli][swap]{f \trr \id}
      &
      X \!\times\! M \dar[kleisli]{f}
      \\
      X \!\times\! M
      \rar[kleisli]{f}
      &
      X
    \end{tikzcd}
\]
  Indeed, the first square commutes because $X\trl \id = \lst_{X,M}$ and the definition of $\lst$. The middle square commutes by definition of Kleisli composition $\kseq$, and the right square commutes since \(f\) is a monoid morphism.
  We therefore get that the following diagram commutes in \(\Kl(\T)\):
  \[
    \begin{tikzcd}[scale cd=.75]
      X \times TM \times TM
      \rar[kleisli, yshift=-2pt, swap]{(X \trl \id) \trr TM}
      \rar[kleisli, yshift=3pt]{X \trl (\id \trr TM)}
      \dar{\id \times \bar{f}^{\#} \times \id}
      \ar[shiftarr={xshift=-60pt}, swap]{dd}{\id \times (\bar{f}^{\#} \times \bar{f}^{\#})}
      \ar[shiftarr={yshift=20pt}, kleisli]{rr}{X \trl \pi_{M, M}}
      &
      X \times M \times TM
      \dar{f \trr TM}
      \rar[kleisli]{(X \times M) \trl \id}
      &
      X \times M \times M
      \drar[kleisli]{X \trl\,m}
      \dar[kleisli]{f \trr \id}
      &
      \\
      X \times (TX)^{X} \times TM
      \rar[kleisli]{\ev \trr TM}
      \dar{\id \times \id \times \bar{f}^{\#}}
      &
      X \times TM
      \dar{\id \times \bar{f}^{\#}}
      \rar[kleisli]{X \trl \id}
      &
      X \times M
      \dar[kleisli]{f}
      &
      X \times M
      \dlar[kleisli]{f}
      \\
      X \times (TX)^{X} \times (TX)^{X}
      \drar[kleisli]{\id \times \kseq}
      \rar[kleisli]{\ev \trr (TX)^{X}}
      &
      X \times (TX)^{X}
      \rar[kleisli]{\ev}
      &
      X
      &
      \\
      &
      X \times (TX)^{X}
      \urar[kleisli]{\ev}
      &
      &
    \end{tikzcd}
  \]
  Here the upper right square uses that \(f\) is central.
  The curried forms of the outer paths precisely yield that \(\bar{f}^{\#}\) is a monoid homomorphism.
\end{proof}

\detailsfor{rem:affine-implies-xx-mon-fg}
\label{proof:rem:affine-implies-xx-mon-fg}
A commutative monad $\T$ is called \emph{affine} if $T 1 \cong 1$. By~\cite[Thm.~2.1]{kock-71}, this is equivalent to commutativity of the following triangle for all $X,Y\in \Set$, where $p_1$ and $p_2$ are the product projections:
\begin{equation}\label{eq:affine} \begin{tikzcd}
TX\times TY \ar[equal]{dr}[swap]{} \ar[tail]{r}{\pi_{X,Y}} & T(X\times Y) \ar[two heads]{d}{\langle Tp_1,Tp_2\rangle} \\
& TX\times TY
\end{tikzcd}
\end{equation}
For every finite set $X$, we have $(X\to X) \cong \prod_{x\in X} X$ and $(X\kleislito X)\cong \prod_{x\in X} T X$ and so repeated application of \eqref{eq:affine} yields a  commutative triangle of the following form, where $\xi_{0}\colon (X\to X)\to (X\kleislito X)$ is the monoid morphism given by $\xi_{0}(f)=f\seq \eta_X$ and $\xi = \xi_{0}^{\#}$:
\begin{equation*}
\begin{tikzcd}
(X\kleislito X) \ar[equal]{dr}[swap]{} \ar[tail]{r}{\lambda_{X, Y}} & T(X\to X) \ar[two heads]{d}{\xi_{X, Y}} \\
& (X\kleislito X)
\end{tikzcd}
\end{equation*}
Thus $\xi$ is surjective, and it is a monoid morphism because $\T$ is commutative.

\end{document}